\documentclass[letterpaper,11pt]{article}
\usepackage{graphicx} % Required for inserting images
\usepackage[]{amsmath,amssymb,amsfonts,latexsym,amsthm,enumerate,fullpage,xcolor,cite,bbm,mathtools}
\usepackage[pagebackref,colorlinks,citecolor=blue,urlcolor=magenta,linkcolor=magenta,,bookmarks=true]{hyperref}
\usepackage[nameinlink, noabbrev, capitalize]{cleveref}
\usepackage{amsmath}
\usepackage{pgfplots}
\title{Linear Hashing Is Optimal}
\author{Michael Jaber\thanks{Department of Computer Science, University of Texas at Austin. \href{mailto:mjjaber@cs.utexas.edu,vmkumar@cs.utexas.edu}{\texttt{\{mjjaber, vmkumar\}@cs.utexas.edu}}. Supported by NSF Grant CCF-2312573 and a Simons Investigator Award (\#409864, David Zuckerman).}
\and 
Vinayak M.\ Kumar\footnotemark[1]
  \and David Zuckerman\thanks{Department of Computer Science, University of Texas at Austin. \href{mailto:diz@cs.utexas.edu} {\texttt{diz@cs.utexas.edu}}. Supported in part by NSF Grant CCF-2312573 and a Simons Investigator Award (\#409864).}
}
  
\date{}
\renewcommand{\backref}[1]{}

\renewcommand{\backrefalt}[4]{%
\ifcase #1 %
\or
[p.\ #2]%
\else
[pp.\ #2]%
\fi}
\newtheorem{theorem}{Theorem}
\newtheorem{lemma}{Lemma}

\newtheorem{claim}{Claim}
\newtheorem{fact}{Fact}

\newtheorem{definition}{Definition}

\newtheorem{remark}{Remark}

\usepackage{comment}
\usepackage{nicefrac}
\crefname{prop}{Proposition}{Propositions}
\crefname{ineq}{inequality}{inequalities}
\creflabelformat{ineq}{#2(#1)#3}

\crefname{THM}{Theorem}{Theorems}

\newcommand{\E}{\mathbb{E}}
\newcommand{\N}{\mathbb{N}}

\newcommand{\F}{\mathbb{F}}

\newcommand{\eps}{\varepsilon}
\newcommand{\Span}{\text{Span}}
\newcommand{\opt}{\textsc{opt}}
\renewcommand{\subset}{\subseteq}

\newcommand{\cE}{\mathcal{E}}

\newcommand{\cK}{\mathcal{K}}

\newcommand{\cH}{\mathcal{H}}

\begin{document}

\maketitle

\abstract{
%We prove that hashing $n$ balls into $n$ bins via random $\F_2$-linear maps yields expected maximum load $O(\log n / \log \log n)$, resolving an open question of Alon, Dietzfelbinger, Miltersen, Petrank, and Tardos (STOC ’97, JACM ’99). Our proof uses potential functions to detect heavy bins.
%More generally, we show that the maximum load exceeds $r\cdot \log n/\log\log n$ with probability at most $O(1/r^2)$. 

%We prove that hashing $n$ balls into $n$ bins via a random matrix over $\F_2$ yields expected maximum load $\Theta(\log n / \log \log n)$, resolving an open question of Alon, Dietzfelbinger, Miltersen, Petrank, and Tardos (STOC ’97, JACM ’99). More generally, we show that the maximum load exceeds $r\cdot \log n/\log\log n$ with probability at most $O(1/r^2)$. Our proof uses potential functions to detect heavy bins.

We prove that hashing $n$ balls into $n$ bins via a random matrix over $\F_2$ yields expected maximum load $O(\log n / \log \log n)$. This matches the expected maximum load of a fully random function and resolves an open question posed by Alon, Dietzfelbinger, Miltersen, Petrank, and Tardos (STOC ’97, JACM ’99). More generally, we show that the maximum load exceeds $r\cdot \log n/\log\log n$ with probability at most $O(1/r^2)$.

%Consider vector spaces $U$ and $V$ over $\F_2$, where $|U|\ge |V|=n$. We show that for any subset $S\subset U$ of size $n$, a random linear map between $U$ and $V$ will map at most $O\left(\log n/\log\log n\right)$ elements of $S$ to the same vector in $V$ with high probability. 
%Consider hashing $n$ balls to $n$ bins. We demonstrate that the simple hash family of random $\F_2$-linear maps, introduced by Carter and Wegman (STOC '77), has expected maximum load $O(\log n/\log\log n)$. This resolves a question of Alon, Dietzfelbinger, Miltersen, Petrank, and Tardos (STOC '97, JACM '99), and is the first example of a universal hash family to achieve optimal expected max-load. We establish this result by carefully designing potential functions that detect when a heavy bin is present and analyzing its growth.
%
}
\section{Introduction}
In 1997, Alon, Dietzfelbinger, Miltersen, Petrank, and Tardos \cite{alon99} asked ``Is Linear Hashing Good?'' We answer this question in the affirmative over $\F_2$.

\begin{theorem}
\label{thm:startthingswithaBANG}
    Let $u\ge \ell\ge 1 $ be integers, $n\coloneqq 2^\ell$, and $\cH$ the set of linear maps $\F_2^u\to\F_2^\ell$. For any $S\subset \F_2^u$ with cardinality $n$, we have 
    %\[\Pr_{h\sim \cH}[\forall y\in \F_2^\ell, |h^{-1}(y)\cap S| \le \frac{48}{\eps} \opt(m,n)]\ge 1-\eps.\] 
    %\[\Pr_{h\sim \cH}\left[\exists y\in \F_2^\ell: |h^{-1}(y)\cap S| \ge C\frac{\log n}{\log\log n} \right] \le \eps.\]
    \[\E_{h\sim \cH}\left[\max_{y\in \F_2^\ell} |h^{-1}(y)\cap S|\right] \le \frac{16\log n}{\log\log n}.\]
\end{theorem}

  Due to the classical balls-and-bins result, this shows that the expected maximum load of a random linear map is within a constant factor of the maximum load of a fully random function. Our proof uses potential functions, which more generally allows us to show that the maximum load exceeds $\frac{r\log n}{\log\log n}$ with probability $O(1/r^2)$.  
  We now provide context on our work, and cover prior results on linear hashing.

\subsection{Motivation}

Consider tossing $n$ balls uniformly and independently into $n$ bins. The maximum number of balls in any bin--the max load--is a well studied quantity critical to randomized algorithm design \cite[Chapter 5]{mitzupfal}. A common application arises in hashing with chaining. In this scenario, keys are mapped to addresses via a random hash, and each address holds a linked list of all keys mapping to it. Consequently, retrieving a key might require sweeping through the largest linked list, and so the worst-case retrieval time would be the max-load. 
%The expected max-load of a hash family is therefore the expected retrieval-time when the family is used to hash with chaining.

A classical result says that tossing $n$ balls randomly into $n$ bins will have expected max-load $O\left(\frac{\log n}{\log\log n
}\right)$, implying that a fully random hash will have only $O\left(\frac{\log n}{\log\log n
}\right)$ expected retrieval time. However, placing balls independently and uniformly requires sampling a truly random function, which has near-maximal time and space complexity. This raises a central question: \begin{center}\emph{What is the simplest hash family with optimal expected max-load?}\end{center}

An engineering approach to this question is to use any tools necessary to design hash functions that minimize the time and space complexity, say in the word-RAM model. Such an approach will naturally give low complexity hash functions, but at the cost of a more contrived construction that is hard to implement practically. Nevertheless, this approach has been extremely successful, and arguably started with Wegman and Carter's definition of $k$-wise independent hash functions \cite{WC81}. One can easily verify that $O(\log n/\log\log n)$-wise independent hash functions have optimal max-load, and only require $O(\log n/\log\log n)$ evaluation time and $O(\log^2 n/\log\log n)$ description size.
Recent works greatly optimized the use of $k$-wise independence, culminating in a construction of Meka, Reingold, Rothblum, and Rothblum \cite{crsw13,mrrr-fastloadbalance} only needing $O(\log n\log\log n)$ bits to describe and $O((\log\log n)^2)$ time to evaluate. This gets us within a $\text{poly}(\log\log n)$ factor of optimal in both parameters. However, these function rely on concatenating hashes of gradually increasing independence, which will require implementing either prime fields or polynomial rings, and performing $\omega(1)$ multiplications over them. Although possible, it is quite complicated and slow to do in practice.

The approach we adopt in this paper is to only consider simple and practical hash functions and analyze their properties as well as possible. To this end, Chung, Mitzenmacher, and Vadhan \cite{cmv13-simplehashingonentropicsources} showed that basic universal hash functions (e.g., linear congruence or multiplication schemes) can achieve optimal max-load if the balls are assumed to have high enough entropy, but say nothing about a worst-case choice of balls. In terms of worst-case results, a surprising result of P\v{a}tra\c{s}cu and Thorup \cite{PT12-tabulation} show that tabulation hashing has optimal max-load. This scheme is only 3-wise independent, is simple and practical to use, and has $O(1)$ evaluation time. However, it requires $O(n^\eps )$ bits to describe.
\subsubsection*{Linear Hash Functions}

%We note that the above evaluation times are in the word-RAM model, thereby hiding the complexity of arithmetic operations over $O(\log n)$-bit integers, like multiplication. However, this assumption becomes unrealistic for large input values. Indeed, Carter and Wegman cited this phenomenon as a main motivator behind the their universal hash constructions \cite{cw79}. All of the above constructions mentioned do involve multiplication, whose best known implementations require $\omega(\log n)$ basic bit operations. \vinny{fix using "How Fast Can We multiply large integers using computers" and motivate further with incremental crytpography, say we only use word additions, so in computationals where addition is better, this is great. Extendability, I can just add columns}

In this paper, we consider an extremely simple hash family proposed in the first paper on universal hashing \cite{cw79}: random matrices over $\F_2$. In particular, let $\ell \coloneqq \log n$, and say we arrange our $n$ bins into a vector space $\F_2^\ell$. Further, arrange the universe set into a vector space~$\F_2^u$. Our hash family is simply the set of linear maps $h:\F_2^u\to\F_2^\ell$.  This family is arguably the easiest to implement over all constructions mentioned thus far. For $u = O(\ell)$, this family only needs to bitwise XOR $O(\log n)$ words together, and takes $O(\log^2 n)$ bits to describe. This family requires no multiplication, is not even 3-wise independent, and can be seen as the simplest version of tabulation hashing, where the lookup tables only stores two values. 

Furthermore, linear hash functions can batch compute keys that are clustered together, which happens often in practice. This notion of efficiency, dubbed \emph{incremental cryptography}, was first introduced in a pair of papers by Bellare, Goldreich, and Goldwasser \cite{BGG94, BGG95}. The motivation behind incremental cryptography is that data to be hashed often consists of slight modifications of each other, whether it be consecutive frames of video footage or edited versions of files. Ideally, after computing the first hash, future hashes should have computation time proportional to the amount of modification made to this initial key. This is indeed the case for linear maps. Let $h(x) \coloneqq \sum_{i:x_i=1}c_i$ be a linear map whose matrix has columns $c_1,\dots, c_u$. If $x$ has Hamming weight $w$, $h(x)$ is simply a sum of $w$ column vectors. Hence, if $x, x^1,\dots x^{100}$ are keys such that each $x^i$ differs from $x$ in less than $w\ll u$ bits, we can compute all the hashes by computing $h(x)$, $ h(x+x^i)$ for all $i$, and then computing $h(x^i) = h(x)+h(x+x^i)$. This gives a total of $u+100w$ additions, which is better than the naive way of computing $h(x), h(x^1),\dots, h(x^{100})$ directly, which can take up to $100u$ additions.

Computational considerations aside, it is a fundamental question to ask how much a random matrix behaves like a random function. From a technical standpoint, dealing with the correlations between linearly dependent balls has proven to be challenging. For example, the mere existence of good linear seeded extractors is a longstanding open problem in pseudorandomness \cite[Question 7.6]{wootters2014thesis}, and highlights our lack of understanding of random linear maps.

%It may seem that because $\cH$ takes $O(\log n)$ time to evaluate and $O(\log^2 n)$ bits to store, the simple solution of $O(\log n/\log\log n)$-wise independent hash functions gives better parameters in both complexity measures. This is indeed the case under the standard assumption that multiplication costs unit time. However, as stated in Carter and Wegman \cite{cw79}, this assumption becomes unrealistic when the input universe set grows large. Without this assumption, $O(\log n/\log\log n)$-wise independence will now require $\Omega(\log^2 n\log\log n)$ evaluation time whereas a $\cH$ will still have evaluation time $O(\log n)$. All other examples mentioned above having better time and space complexity in the word RAM model have a similar blowup in runtime when accounting for multiplication \cite{crsw13,mrrr-fastloadbalance}.

%\vinny{Add pseudorandomness motivation here? } 
% Talk about linear-seeded extractors, get a sense for what is actually known, mention that LHL gives something with terrible seed length

\subsubsection*{The Expected Max-Load of Linear Maps}

Characterizing the max-load of $\cH$, the set of linear maps, has remained elusive. As a preliminary bound, the expected max-load of $\cH$ is $O(\sqrt{n})$ since it is universal \cite{cw79}. Indeed, some universal hash functions saturate this bound \cite{alon99}. However, Markowsky, Carter, and Wegman \cite{mcw78} showed $\cH$ achieves expected max-load $O(n^{1/4})$, significantly outperforming the generic universal bound \cite{alon99}. This initiated further study of $\cH$'s max-load, with Mehlhorn and Vishkin \cite{mv81} (implicitly) showing a subpolynomial bound of $2^{O(\sqrt{\log n})}$, and later with Alon, Dietzfelbinger, Miltersen, Petrank, and Tardos \cite{alon99} giving a bound of $O(\log n\log\log n)$. Since then, progress has largely stalled, except for a note by Babka \cite{babka18}, who observed that an improved choice of parameters in the argument of \cite{alon99} yields a bound of $O(\log n)$. Unfortunately, the argument in \cite{alon99} has an inherent barrier at $O(\log n)$,\footnote{Applying their argument to a purely random function only yields an $O(\log n)$ expectation bound on the max-load.} leaving open whether $\cH$ can match the performance of a fully random map. Indeed, this natural question was explicitly posed in \cite{alon99}.

% With regards to the max-load of $\cH$, the fact that $\cH$ is universal already guarantees an expected max-load of $O(\sqrt{n})$ \cite{cw79}, and indeed some universal hash functions saturate this bound \cite{alon99}. However, Markowsky, Carter, and Wegman \cite{mcw78} showed $\cH$ achieves expected max-load $O(n^{1/4})$, performing much better than the generic universal bound \cite{alon99}. This initiated work studying the max-load of $\cH$, with Mehlhorn and Vishkin \cite{mv81} (implicitly) showing a bound of $O(2^{\sqrt{\log n}})$, and later with Alon, Dietzfelbinger, Miltersen, Petrank, and Tardos \cite{alon99} giving a bound of $O(\log n\log\log n)$. Babka \cite{babka18} observed that a better setting of parameters in the argument of \cite{alon99} yields $O(\log n)$. Unfortunately, the argument in \cite{alon99} has an inherent barrier at $O(\log n)$,\footnote{Applying their argument to a purely random function only yields an $O(\log n)$ expectation bound on the max-load.} and left open whether $\cH$ could perform as well as a fully random map. Indeed, this natural question was explicitly posed in \cite{alon99}. 

\subsection{Our Results}

We fully resolve this question by showing that a random linear map hashing $n$ balls to $n$ bins will have expected max-load $O\left(\frac{\log n}{\log\log n}\right)$. Our proof easily generalizes to $m$ balls and $n$ bins for $m\neq n$. In particular, define \[\opt(m,n) = \begin{cases}\frac{\log n}{\log\left(\frac{n\log n}m\right)} & m\le \frac{1}2n\log n \\  \frac{m}n & m > \frac{1}2 n\log n.\end{cases}\] It is well known that a random function mapping $m$ balls to $n$ bins will have expected max-load $\Theta(\opt(m,n))$ \cite[Theorem 1]{RS98-ballsbinstightanalysis}. We show that a random linear map performs just as well. 

\begin{theorem}[\Cref{thm:startthingswithaBANG} generalized]
\label{thm:balls-and-bins}
    Let $u\ge \ell, m$ be integers, $n\coloneqq 2^\ell$ , and $\cH$ the set of linear maps $\F_2^u\to\F_2^\ell$. For any $S\subseteq \F_2^u$ with cardinality $m$, 
    %\[\Pr_{h\sim \cH}[\forall y\in \F_2^\ell, |h^{-1}(y)\cap S| \le \frac{48}{\eps} \opt(m,n)]\ge 1-\eps.\] 
    %\[\Pr_{h\sim \cH}[\exists y\in \F_2^\ell: |h^{-1}(y)\cap S| > (2+3/\eps) \cdot \opt(m,n)] < \eps.\] 
    \[\E_{h\sim \cH}\left[\max_{y\in \F_2^\ell} |h^{-1}(y)\cap S|\right] \le 16\cdot \opt(m,n)\]
\end{theorem} 

Our \cref{thm:balls-and-bins} is a simple corollary of the following theorem, which gives quadratically decaying tail bounds on the max-load.

\begin{theorem}
     Let $u\ge \ell\ge 1, m\ge 1$ be integers, $n\coloneqq 2^\ell$, and $\cH$ the set of linear maps $\F_2^u\to\F_2^\ell$. For any $S\subset \F_2^u$ with cardinality $m$ and $r\ge 6$, 
    %\[\Pr_{h\sim \cH}[\forall y\in \F_2^\ell, |h^{-1}(y)\cap S| \le \frac{48}{\eps} \opt(m,n)]\ge 1-\eps.\] 
    %\[\Pr_{h\sim \cH}[\exists y\in \F_2^\ell: |h^{-1}(y)\cap S| > (2+3/\eps) \cdot \opt(m,n)] < \eps.\] 
    \[\Pr_{h\sim \cH}\left[\max_{y\in \F_2^\ell} |h^{-1}(y)\cap S|\ge r\cdot \opt(m,n)\right] \le \frac{49}{(r-2)^2}.\]
\end{theorem}

Our techniques arguably yield a more streamlined proof than \cite{alon99}, which first proves a coupon collector property, and then cleverly converts it to a max-load bound. In contrast, we directly analyze the max-load by defining a potential function and tracking its growth. We elaborate further in \Cref{sec:tech}.

%When $m=n$, we also show the expected max-load of $\cH$ is $O\left(\frac{\log n\log\log\log n}{\log\log n}\right)$, which improves on the previous expectation bound of $O(\log n)$ \cite{alon99,babka18}. This brings us a $\log\log\log n$ factor away from the optimal quantity, which would answer a question of \cite{alon99}.

In terms of matching lower bounds, we note \Cref{thm:balls-and-bins} is tight up to constant factors. In particular, Celis, Reingold, Segev, and Wieder show that \emph{any} family of hash functions will have max-load  $\Omega(\opt(m,n))$ with high probability \cite[generalization of Theorem 5.1]{crsw13}. 

%When $m = O(n\log n)$, we show that most sets of $m$ balls have \emph{no} linear map that achieves a max-load better than $\Theta(\opt(m,n))$ with high probability. Here, we show

%\begin{theorem}
%    Let $m,n$ be integers, and define $\ell \coloneqq \log n$. There exists a set $S\subset\F_2^u$ of cardinality $m$ such that for every linear map $h:\F_2^u\to\F_2^\ell$, there exists $y$ such that $|h^{-1}(y)\cap S|\ge (1-o(1))\opt(m,n)$. 
%\end{theorem}

%The proof of this uses a standard probabilistic method argument and uses techniques from the fully random balls-and-bins setting.
Interestingly, in the regime $m =  \Omega(n\log n)$, we show that with high probability, \emph{every} bin has load within a constant factor of $\opt(m,n)$.

\begin{theorem}
\label{thm:our-two-sided}
     Let $u\ge \ell\ge 1, m\ge 1$ be integers, $S\subset\F_2^u$ a set of cardinality $m$, and $\cH$ the set of linear maps $\F_2^u\to\F_2^\ell$. For every $0 < \eps < 1/2$, there exists constants $C_1 < C_2$ depending on $\eps$ such that for $m \ge C_1^{-1}n\log n$,
     \[\Pr_{h\sim \cH}\left[\forall y\in \F_2^\ell,~ C_1 \frac{m}n\le |h^{-1}(y)\cap S|\le C_2 \frac{m}n\right] \ge 1-\eps.\] In particular, $C_1 = \Omega(\eps^{74})$ and $C_2 = O(\eps^{-1/2})$.
     %\[\Pr_{h\sim \cH}\left[\forall y\in \F_2^\ell,~ C_1 \opt(m,n)\le |h^{-1}(y)\cap S|\le C_2 \opt(m,n)\right] \ge 1-\eps.\]
\end{theorem}
This can be seen as a generalization of the result of \cite[Theorem 7]{alon99}, which states that $\cH$ satisfies the \emph{covering property}. The covering property, named after the fact that the cover time of a random walk on the complete graph on $n$ vertices is $O(n\log n)$, says that if $\Omega(n\log n)$ balls are mapped to $n$ bins by $h\sim \cH$, every bin will be occupied with high probability. We show these hash functions actually satisfy a  \emph{blanketing property}: if $\Omega(n\log n)$ balls are mapped to $n$ bins by $h\sim \cH$, every bin will contain $\Omega(\log n)$ balls with high probability. A fully random $h$ has this property by the fact that the blanket time \cite{WinklerZuckerman1996} of the complete graph is $O(n\log n)$. Interestingly, we use potential functions to prove this claim as well.

\subsubsection*{Comparison with Dhar and Dvir}

A recent work of Dhar and Dvir \cite{dhardvir} also established two-sided bounds on all bins for linear hash functions. They proved the following: 

\begin{theorem}[\hspace{1sp}\cite{dhardvir}, Theorem 2.4]\label{thm:dhardvir}
    Let $u\ge \ell\ge 1$ be integers, $n\coloneqq 2^\ell$, $\eps,\tau\in (0,1)$, and $S \subset \F_2^{u}$ a set of cardinality $m$. Let $\cH$ be the set of linear maps $h:\F_2^u\to\F_2^\ell$. For $m = \Omega(n \log^4 (n/\tau\eps))$ we have
    \[\Pr_{h\sim \cH}\left[\forall y\in \F_2^\ell,~ (1-\tau)\frac{m}{n} \leq |h^{-1}(y) \cap S| \leq (1+\tau)\frac{m}{n}\right] \ge 1-\eps.\]
    
\end{theorem}
%There are few differences to note between \Cref{thm:our-two-sided} and \Cref{thm:dhardvir}. The most obvious difference is that Dhar and Dvir give a load guarantee for \emph{every} bucket, whereas our result controls the load of the heaviest bucket. They also establish two sided \emph{additive error} bounds to the optimal load, whereas we only establish a constant multiplicative error bound. However, Dhar and Dvir's results only hold in the regime $m = \Omega(n\log^4 n)$, whereas we can handle $m$ as small as $\Theta(n\log n)$. This gap between $m$ and $n$ is optimal, as it was shown in ~\cite{alon99} that there exists sets of size $.6n\log n$ such that no linear map will even have all buckets occupied.
 Take $\eps = O(1)$. We note that \Cref{thm:dhardvir} gives strong deviation bounds from the optimal load, especially if one sets $\tau = o(1)$, whereas our \Cref{thm:our-two-sided} only establishes that all bins are within some constant multiplicative factor of optimal. However, \Cref{thm:dhardvir} only holds in the regime $m = \Omega(n\log^4 n)$, whereas \Cref{thm:our-two-sided} can handle $m$ as small as $\Theta(n\log n)$. This gap between $m$ and $n$ is optimal, as it was shown in ~\cite{alon99} that there exists sets of size $0.69n\log n$ such that no linear map will even occupy all bins. Our techniques are also quite different. Dhar and Dvir \cite{dhardvir} first reduce the problem of bounding the max-load to bounding the size of Furstenberg sets and then apply the polynomial method, while our proof constructs a potential function and bounds its growth to yield the result. 

In summary: 
\begin{itemize}
    \item when $m = O(n\log n)$, \Cref{thm:balls-and-bins} gives us optimal upper bounds (up to a constant multiplicative factor) on all bins,
    \item when $m\in [\Omega(n\log n), O(n\log^4 n)]$, \Cref{thm:our-two-sided} gives us two-sided multiplicative bounds from the mean load on all bins,
    \item and when $m = \Omega(n\log^4 n)$, \Cref{thm:dhardvir} gives very strong bounds on the additive deviation of all bins from the mean.
\end{itemize}

\subsection{Proof Overview}
\label{sec:tech}

For simplicity, assume $\cH$ is a random \emph{surjective} linear map $\F_2^u\to \F_2^\ell$. We would like to directly analyze the load distribution of $\cH$, but this is quite complicated to do. Instead, for a set of balls~$S$ and hash $h$, we define a potential function $\Phi$ that measures how ``imbalanced'' the allocation of $S$ by $h$ is. In particular, we want a potential function $\Phi \coloneqq \Phi(S,h)$ such that if one preimage of $h$ contains $\ge t$ elements in $S$, then $\Phi \ge f(t)$ for some function $f$. The hope is now that analyzing $\Phi$ will be much easier than analyzing the load distribution directly.

 We use the potential $\E_y[b^{|h^{-1}(y)\cap S|}]$, which takes the average of the exponentials of all the bin loads with some base $b>1$. To analyze this potential, we think of hashing our universe, $\F_2^u$, ``one kernel vector at a time.'' In particular, any $h:\F_2^u\to\F_2^\ell$ can be decomposed into $h = h_1\circ h_2\circ\cdots \circ h_{u-\ell}$, where $h_i:\F_2^{u-i+1}\to \F_2^{u-i}$ is a random surjective map. Each $h_i$ is much simpler to analyze, and allows us to show that if $h_{\le i} \coloneqq h_1\circ \cdots \circ h_i$ and $\Phi_i \coloneqq \Phi(S,h_{\le i})$, then $\E[\Phi_{i+1}|\Phi_1,\dots \Phi_i]\le \Phi_i^2$.

It remains to show that these conditional expectation bounds can be leveraged to give a tail bound on $\Phi_{u-\ell} = \Phi$. A tempting approach is to use the conditional expectations to bound $\E[\Phi] \le f(t)/2$, and then apply Markov's inequality to get a tail bound. Unfortunately, there exists random variables satisfying the conditional expectation bounds, but with $\E[\Phi]$ much larger than $f(t)$. Nevertheless, we prove a technical lemma showing that although $\E[\Phi]$ may be larger than $f(t)$, the conditional expectation bounds enforce that $\Phi$ \emph{will still be smaller than $f(t)$ with good probability}. 
%Our proof is heavily motivated by and iterates the ``geometric'' proof of Markov's inequality,\footnote{See Theorem 1.1 of \url{http://courses.cms.caltech.edu/cs139/notes/lecture01.pdf}} \mj{what are your thoughts on the reviewer's comment on this link? i looked at the link and the proof you are mentioning is pretty straightforward/correct so im okay either way}\vinny{I think the reviewer's concern is valid, i just havent been able to find a literature citation with this proof :(}\mj{it just occurred to me that if the link gets taken down that is a bigger problem, what if we just write it down somewhere in the paper since it is short?}
%We believe this lemma may be of independent interest.
A tail bound of similar flavor was established in \cite{alon99}, but only works on random variables less than 1, and is inapplicable to our setting.

This lemma allows us to prove that the max-load exceeds $\frac{r\log n}{\log\log n}$ with probability $O(1/r)$. Unfortunately, this is not enough to deduce an $O(\log n/\log\log n)$ bound in expectation. In fact, only using the property $\E[\Phi_{i+1}|\Phi_1,\dots \Phi_i]\le \Phi_i^2$ provably cannot give a stronger tail bound. Thankfully, our potential functions have a strong monotonicity property: $\forall i, \Phi_{i+1}-1\ge 2(\Phi_i-1)$. With some technical work, we can leverage the monotonicity and squared conditional expectation properties to obtain quadratically stronger tail bounds, from which optimal expected max-load follows straightforwardly.
%Using the monotonicity property, squared conditional expectation property, and some technical work, we can then derive quadratically stronger tail bounds, which easily imply optimal expected max-load.

Surprisingly, our potential functions also allow us to establish \emph{lower bounds} on the loads of all bins. By setting the base $b<1$, the potential now detects whether some bin is light. Combining this with the max-load analysis allows us to deduce our two-sided bounds on the bin loads.

To the best of our knowledge, our analysis provides the first proof of optimal max-load for a function that is universal, but not 3-wise independent. In fact, our proof technique applies to a broader class of hash functions mapping $[2^u]\to[2^\ell]$. Imagine starting with $2^u$ bins, where each of the $2^u$ universe elements are in their own bin. For $u-\ell$ iterations, pseudorandomly pair bins up in a pairwise independent fashion (i.e., for an arbitrary bin, the marginal distribution of the bin's partner is uniform among all bins) and consolidate each pair into one bin. Consequently, each iteration halves the number of bins, and at the end of this process, we will have hashed into $n$ bins. Our techniques show that any such hashing scheme will have optimal expected max-load. The case of surjective linear maps is when in each round, a random vector $v$ is picked, and each bin $x$ is paired with bin $x+v$.
%We note that the proof of \cite{alon99} can be viewed as using a degenerate case of our potential functions. When proving that $\cH$ has the covering property, \cite{alon99} analyzes the decay rate of the density of empty buckets. The density of empty buckets is simply the exponential potential function with base $b=0$ (with the convention $0^0 = 1$).

%\subsection*{Organization} We give preliminaries \Cref{sec:prelim}. We then define the potential functions and prove their relevant properties in \Cref{sec:assocfuncpotential}. As a warm-up, we will then give a short proof that random linear maps have optimal max-load with probability $.99$ in \Cref{sec:whp}. This will capture the main intuition behind the proof of optimal expected max-load in \Cref{sec:expectation}. Finally, we show our two-sided bounds in \Cref{sec:two-sided}.
\section{Preliminaries}
\label{sec:prelim}

\subsection{Notation}
We let $\log$ denote the base-2 logarithm, and $\ln$ denote the base-$e$ logarithm. For a list of vectors $v_1,\dots, v_n$, we will denote the tuple $v_{\le i} \coloneqq (v_1, v_2, \dots, v_i)$ and $V_i = \Span(v_{\le i})$ (by convention, $V_0 = \{0\}$). We will use the convention $0^0=1$. Throughout this paper, we always have $\ell = \ell(n) \coloneqq \log n$ and $n$ a power of 2. $\F_2$ is the finite field over $\{0,1\}$. For a set $S$, we write $s\sim S$ to denote that $s$ is sampled uniformly at random from $S$.

\begin{definition}
\label{defn:maxload}
Let $h: A\to B$ and $S\subset A$ be a subset. We define the maximum load function \[M(S,h) \coloneqq  \max_{y\in B}|h^{-1}(y)\cap S|.\] 
\end{definition}
%If $A = B$, then we will define $h^{(m)} = h\circ h\circ\cdots \circ h$ to be $m$-fold composition of $h$ with itself (with $h^{(0)}$ being the identity function).

\subsection{Inequalities}
We will use the classic inequality $1+x\le e^x$ for all $x$, and the following not-so-classic variant.
\begin{fact}
\label{fact:expineq}
    For $x< 1$, we have $1-x\ge e^{-\frac{x}{1-x}}$.
\end{fact}
\begin{proof}
    Note $\frac{1}{1-x} = 1 + \frac{x}{1-x}\le e^{\frac{x}{1-x}}$. When $x<1$, both sides of the inequality are positive, and taking reciprocals gives the fact.
\end{proof}

We will also use Bernoulli's Inequality.
\begin{fact}
    \label{claim:silly}
    For an integer $n\ge 1$ and $1+x\ge 0$, $(1+x)^n\ge 1+nx$.
\end{fact}

Finally, for the two-sided bounds, we will use standard facts about martingales, such as the Doob martingale and an elementary version of Azuma's inequality \cite[Chapter 12]{mitzupfal}. 

\begin{fact}
\label{fact:azuma}
    Let $(X_i)_{i\ge 0}$ be a martingale such that for all $i$, $|X_{i+1}-X_i|\le 1$. Then for all positive integers $k$ and real number $\eps > 0$, we have \[\Pr[X_k-X_0\le -\eps]\le e^{-\eps^2/2k}.\]
\end{fact}

\subsection{Random Linear Maps}
If $h:\F_2^u\to\F_2^\ell$ is a linear map with kernel $V\leq \F_2^u$ and $y$ is in the image of $h$, then $h^{-1}(y) = x+V$ for some $x\in \F_2^u$. A random \emph{surjective} linear map $h: \F_2^u\to\F_2^\ell$ is equivalent to sampling a uniform $(u-\ell)$-dimensional subspace $V\leq \F_2^u$, and then sampling a uniform linear $h$ with kernel $V$. For any $u\ge t\ge \ell$ and surjective map $h_2:\F_2^t\to \F_2^\ell$, if $h_1:\F_2^u\to\F_2^t$ is a uniform random linear map, then $h_1\circ h_2: \F_2^u\to\F_2^\ell$ is a uniform random linear map.

\section{Introducing the Potential Functions}
\label{sec:assocfuncpotential}
Consider $k$ vectors $v_{\le k}\in (\F_2^u)^k$ and a set of balls $S\subset \F_2^u$. Recall $V_i = \Span(v_{\le i})$. To $(v_{\le k}, S)$ we will associate functions $\{S_i: \F_2^u\to\N \}_{0\le i \le k}$, defined by \begin{equation}S_i(x) \coloneqq |(x+V_i)\cap S|.\end{equation} In particular, $S_i$ depends on $S$ and $v_{\le i}$. Intuitively, $S_0$ is (the indicator of) $S$, and $S_i(x)$ is the number of balls in the same bin as $x$ after hashing according to the kernel vectors $v_{\le i}$.

 %For a set of vectors $v_{\le k}\subset \F_2^u$ and a set $S\subset \F_2^u$, we will \emph{associate} functions $\{S_i\}_{0\le i \le k}$ mapping $\F_2^u\to \N$, defined as follows. \begin{align}S_{i}(x) \coloneqq \begin{cases} \mathbf{1}(x\in S) & i = 0 \\  S_{i-1}(x) + S_{i-1}(x+v_i) &  i\ge 1\end{cases}\label{eqn:assoc}\end{align} Intuitively, $S_0$ is simply (the indicator of) $S$, and in general $S_i(x)$ is the number of balls in the same bin as $x$ after hashing according to the kernel vectors $v_{\le i}$.
 
 % More formally, notice if $\{v_i\}_{i\in [k]}$ are linearly independent, we have the following, where $V_i \coloneqq \Span(v_{\le i})$.
%\begin{proposition}\label{clam:Sinduct} If $\{v_{\le i}\}$ are linearly independent vectors in $\F_2^u$, then \[S_i(x) = \sum_{v\in V_i} S_0(x + v) = |(x + V_i)\cap S|.\]\end{proposition}

%\begin{proof}
%    We apply induction on $i$. The base case is true because \[|(x+V_0)\cap S| = |\{x\}\cap S| = {\mathbf 1}(x\in S) = S_0(x).\] Assume the proposition holds true for $i$. We then have for $v_{i+1}\notin V_i$ that \[S_{i+1}(x) = S_i(x) + S_i(x+v_{i+1}) = |(x + V_i)\cap S| + |(x + v_{i+1} + V_i)\cap S| = |(x+V_{i+1}) \cap S|,\] where the last equality follows from the fact $V_i$ and $v_{i+1} + V_i$ partition $V_{i+1}$, since $v_{i+1}\notin V_i$.
%\end{proof}

Now for any real number $b \ge 0$, we can define a sequence of potential functions \begin{equation*}\Phi_i = \Phi_i(S;b; v_{\le i})\coloneqq \E_{x\sim \F_2^u}[b^{S_i(x)}].\end{equation*} The following claim shows that if $b\ge 1$, a heavy bin implies large potential.

\begin{lemma}

\label{clm:potential-load}
\label{clm:load-potential}
    Let $v_{\le i}$ be a fixed tuple of linearly independent vectors in $\F_2^u$, and let $b \ge 1$. If there exists $x\in \F_2^u$ such that $|(x+V_i)\cap S|\ge m$, then $\Phi_i \ge \frac{b^m}{2^{u-i}}$.
\end{lemma}

\begin{proof}
    If $v\in V_i$, then $v+V_i = V_i$. Hence, for all $v\in V_i$, $S_i(x+v) = S_i(x)\ge m$. Therefore, \[\Phi_i = \frac{1}{2^u}\sum_{x'\in \F_2^u} b^{S_i(x')}\ge \frac{1}{2^u}\sum_{x'\in x+V_i} b^{S_i(x')}\ge \frac{1}{2^u}  |V_i|b^m = \frac{b^m}{2^{u-i}},\]
    where we used the fact $b^x$ is increasing for $b\ge 1$.
\end{proof}

Interestingly, if we use a base $b\le 1$, we can detect by the same test if some bin is light as well. This will help us establish two-sided bounds on the bins later on.

\begin{lemma}
\label{clm:blanketing}
Let $v_{\le i}$ be a fixed tuple of linearly independent vectors in $\F_2^u$, and let $b \le 1$. If there exists $x\in \F_2^u$ with $|(x + V_i)\cap S| \le m$, then $\Phi_i \ge \frac{b^m}{2^{u-i}}$.
\end{lemma}

\begin{proof}
    The proof is essentially the same as \Cref{clm:load-potential}. We observe that for all $v\in V_i$, ${S_i(x+v) \le m.}$ Hence, \[\Phi_i = \frac{1}{2^u}\sum_{x'\in \F_2^u} b^{S_i(x')}\ge \frac{1}{2^u}\sum_{x'\in x+V_i} b^{S_i(x')}\ge \frac{1}{2^u}  |V_i|b^m = \frac{b^m}{2^{u-i}},\] where we used the fact $b^x$ is decreasing for $b\le 1$.
\end{proof}

The following claim relates $S_i$ to $S_{i+1}$.

\begin{claim}
\label{clm:induct-defn}
For any vectors $v_{\le (i+1)}$, we have
    \[S_i(x)+S_i(x+v_{i+1}) = \begin{cases}2S_i(x) & v_{i+1}\in V_i \\ S_{i+1}(x) & v_{i+1}\notin V_i\end{cases}.\]
\end{claim}

\begin{proof}
If $v_{i+1}\in V_i$, then $V_i = v_{i+1}+V_i$, and so \[S_i(x) + S_i(x+v_{i+1}) = S_i(x) + S_i(x) = 2S_i(x).\] If $v_{i+1}\notin V_i$, then $V_{i+1} = V_i\sqcup (v_{i+1}+V_i)$. Thus, $(x+V_i)\cap S$ and $(x+v_{i+1}+V_i)\cap S$ partition $(x+V_{i+1})\cap S$, implying that $S_i(x) + S_i(x+v_{i+1}) = S_{i+1}(x).$
\end{proof}
Using the above claim, we can prove the following crucial lemma which upper bounds the conditional expectations of our potentials. 

\begin{lemma}
\label{clm:expsquare}
    Let $v_{\le i}\in (\F_2^u)^i$, and let $v_{i+1}\sim \F_2^u\setminus V_i$. We have \[\E_{v_{i+1} }[\Phi_{i+1} ] \le \Phi_i^2.\]
\end{lemma}

\begin{proof}
    For $x,v_{i+1}\in \F_2^u$ picked uniformly and independently, $x$ and $x+v_{i+1}$ are uniform and independent as well. Hence \begin{equation}\E_{x,v_{i+1}}[b^{S_i(x) + S_i(x+v_{i+1})}] = \E_{x,v_{i+1}}[b^{S_i(x) + S_i(v_{i+1})}] = \E_x[b^{S_i(x)}]\cdot \E_{v_{i+1}}[b^{S_i(v_{i+1})}] = \Phi_i^2.\label{eqn:global-avg}\end{equation}  Now for any fixed $v_{i+1}\in V_i$, we have \begin{equation}\E_{x}[b^{S_i(x) + S_i(x+v_{i+1})}] = \E_x[b^{2S_i(x)}]\ge \E_x[b^{S_i(x)}]^2 = \Phi_i^2\label{eqn:above-avg}\end{equation} by \Cref{clm:induct-defn} and convexity. By an averaging argument, (\ref{eqn:global-avg}) and (\ref{eqn:above-avg}) imply \begin{align*}\Phi_i^2 & \ge \E_{x,v_{i+1}}[b^{S_i(x) + S_i(x+v_{i+1})}|v_{i+1}\notin V_i] \\ & = \E_{x,v_{i+1}}[b^{S_{i+1}(x)}|v_{i+1}\notin V_i] \\ &  = \E_{v_{i+1}}[\Phi_{i+1}|v_{i+1}\notin V_i],\end{align*} where the first equality follows from \Cref{clm:induct-defn}.
\end{proof}

%\begin{proposition}
%\label{clm:expsquare}
%    Let $v_{\le i}$ be linearly independent vectors in $\F_2^u$, and let $W$ be the uniform distribution over $ \F_2^u\setminus V_i$. Then we have \[\E_{v_{i+1} \sim W}[\Phi_{i+1}] \le \Phi_i^2.\]
%\end{proposition}

%\begin{proof}
%For any fixed linearly independent $v_{\le i}$, if $v_{i+1}$ is picked uniformly at random from \emph{all} of $\F_2^u$, \begin{align*}\E_{v_{i+1}}[\Phi_{i+1}]  = \E_{v_{i+1}}[\E_x[b^{S_{i+1}(x)}]]    = \E_{x,v_{i+1}}[b^{S_i(x) + S_i(x+v_{i+1})}]   = \E_x[b^{S_i(x)}]^2  &  = \Phi_i^2,\end{align*} where the penultimate inequality follows from the fact $x$ and $x+v_{i+1}$ are independent random variables. Whenever $v_{i+1}\in V_i$, it follows by \Cref{clam:Sinduct} that \[S_i(x+v_{i+1}) = |(x + v_{i+1} + V_i)\cap S| = |(x + V_i)\cap S|= S_i(x).\] Therefore by \Cref{eqn:assoc}, if $v_{i+1}\in V_i$, then \[\Phi_{i+1} = \E_x [b^{S_{i+1}(x)}] = \E_x [b^{2S_{i}(x)}] \ge \E_x [b^{S_i(x)}]^2 = \Phi_i^2 \] via Jensen's inequality. Consequently by an averaging argument, it must be the case \[\E_{v_{i+1}\sim W}[\Phi_{i+1}] \le \Phi_i^2.\]
%\end{proof}

At a high level, we will use \Cref{clm:expsquare} to upper bound the potential, from which \Cref{clm:potential-load} implies a small max-load.

While the above suffices to get decent tail bounds on the max-load, for technical reasons we will need the following lemma, which is essential to establishing optimal expected max-load and quantitatively stronger two-sided bounds.

\begin{lemma}
\label{prop:strides}
Let $v_{\le i}\in (\F_2^u)^i$, and $v_{i+1}\in \F_2^u\setminus V_i$. For any $b\ge 0$ we have $\Phi_{i+1}-1\ge 2(\Phi_i - 1)$. When $b\le 1$, we also have $\Phi_{i+1}\le \Phi_i$.
\end{lemma}

\begin{proof}

By \Cref{clm:induct-defn}, elementary manipulations, and linearity of expectation,
    \begin{align*}\Phi_{i+1}-1 &  = \E_x[b^{S_i(x) + S_i(x+v_{i+1})}-1] \\ &  = \E_x[b^{S_i(x)}-1]+\E_x[b^{ S_i(x+v_{i+1})}-1] +\E_x[ b^{S_i(x)+S_i(x+v_i)}-b^{S_i(x)} -b^{S_i(x+v_{i+1})} +  1] \\ & = 2\E_x[b^{S_i(x)}-1] + \E_x[( b^{S_i(x)}-1)(b^{S_i(x+v_{i+1})}-1)] \\ & \ge 2(\Phi_i-1),\end{align*}
 where the inequality follows from the fact for any $b,r,s\ge 0$,  $b^r-1$ and $b^s-1$ have the same sign. When $b\le 1$, \Cref{clm:induct-defn} and the fact $f(r) =b^r$ is decreasing tells us \[\Phi_{i+1} = \E_x[b^{S_i(x)+S_i(x+v_{i+1})}]\le \E_x[b^{S_i(x)}] = \Phi_i. \]
 
\end{proof}
 
\section{Warmup: Optimal Max-Load With .99 Probability}
\label{sec:whp}

In this section, we will use the potentials we constructed to show that a random linear map has maximum load $O\left(\frac{\log n}{\log\log n}\right)$ with probability $0.99$.

\subsection{The Existential Case}

 To give intuition on how the potential functions will be used, we first prove a preliminary result. We will show that for any choice of $n$ balls, there \emph{exists} a linear hash map that has a maximum load of $\frac{2\ln n}{\ln\ln n}$. To the best of our knowledge, even this existential result was not known prior to our work. Recall \Cref{defn:maxload}.
\begin{theorem}
\label{thm:hashingexistence}
    Let $u\ge \ell\ge 1$ be integers, $n\coloneqq 2^\ell$, and $S\subset\F_2^{u}$ be of cardinality $n$. There exists a linear map $h: \F_2^{u}\to\F_2^\ell$ such that \[M(S,h) < \frac{2\ln n}{\ln\ln n}.\] 
\end{theorem}

\begin{proof}
Define $k\coloneqq u-\ell$. We will carefully pick linearly independent kernel vectors $v_{\le k}$ from $\F_2^{u}$, and then argue that any linear  $h: \F_2^u\to\F_2^\ell$ with kernel $V\coloneqq \text{Span}(v_{\le k})$ will have small maximum load.

Using the vectors $v_{\le k}$ and set $S$, define $\{S_i(x)\coloneqq |(x+V_i)\cap S|\}_{0\le i\le k}$ and define potentials $\{\Phi_i \coloneqq \E_{x\sim \F_2^u}[b^{S_i(x)}]\}_{0\le i\le k}$ for $b = \ln n$. As $|S| = n$, we can compute \[\Phi_0 = \E_x[b^{S_0(x)}] = \frac{n}{2^u}\cdot b + \left(1-\frac{n}{2^u}\right)\cdot 1\le 1 + \frac{\ln n}{2^k}.\]

By \Cref{clm:expsquare}, for $i = 1,2,\dots, k$ we can iteratively pick $v_{i+1}\notin \text{Span}(v_{ \le i})$ such that $\Phi_{i+1}\le \Phi_i^2$. Upon picking these $k$ linearly independent vectors $v_{\le k}$ in this manner, we claim any $h$ with $\text{Span}(v_{\le k}) \eqqcolon V$ as the kernel will yield the desired result. To see this, first notice \[\Phi_k \le  \Phi_{k-1}^2 \le  \cdots \le \Phi_0^{2^k} \le \left(1+\frac{\ln n}{2^k}\right)^{2^k} < e^{\ln n} = n.\] Assume there was $y$ such that $|h^{-1}(y)\cap S| \ge \frac{2\ln n}{\ln\ln n} $. Then there must exist $x$ such that ${|(x+V)\cap S|}\ge  \frac{2\ln n}{\ln\ln n}$. But by \Cref{clm:potential-load}, this implies \[\Phi_k \ge  \frac{(\ln n)^{\frac{2\ln n}{\ln \ln n}}}{2^{u-k}} =  \frac{n^2}{2^\ell} = n,\] absurd!
\end{proof}

\subsection{A Tail Bound to Boost Existence to Abundance} 

The above theorem showed the existence of a choice of kernel vectors $v_1,\dots, v_k$ that minimized the potential function, thereby implying the existence of a load-balancing linear map. We would like to show most choices of $v_{\le k}$ will minimize the potential. To do so, we will prove a technical lemma that converts the conditional expectation upper bound guarantees from \Cref{clm:expsquare} into tail bounds on the final potential. In particular, while the last section showed the existence of $v_{\le k}$ such that $\Phi_k\le \Phi_0^{2^k}$, the following lemma shows that most choices of $v_{\le k}$ will have $\Phi_k\le (\Phi_0^{2^k})^{O(1)}$.

\begin{lemma}
\label{lem:potential-tail-bound}
     Let $X_0\ge 1$ be a constant, and let $X_1,\dots, X_k\ge 1$ be random variables satisfying $\E[X_{i+1}|X_{\le i}] \le X_i^2$. For any $t > 1$,  \[\Pr[X_k\ge t^{2^{k-1}}]\le \frac{X_0^2 -1}{t-1}.\]
\end{lemma}

\begin{remark}
\label{rem:initial-tail-tightness}
    \Cref{lem:potential-tail-bound} is tight for the sequence \begin{itemize}\item $X_1  = \begin{cases} t & \text{with probability }\frac{X_0^2-1}{t-1} \\ 1 & \text{otherwise} \end{cases}$,
    \item For $i > 1$, $X_{i+1} = X_i^2.$ \end{itemize} 
\end{remark}

\begin{proof}[Proof of \Cref{lem:potential-tail-bound}]
We proceed by induction on $k$. For $k=1$, it follows by Markov's inequality and the fact $X_1-1$ is a nonnegative random variable that \[\Pr[X_1\ge t] = \Pr[X_1 - 1 \ge t - 1]\le \frac{X_0^2-1}{t-1}.\] Now assume we have the lemma for $k$ and wish to prove it for $k+1$. We can bound \begin{align*}\Pr[X_{k+1}\ge t^{2^k}]  & = \E_{X_1}\left[\Pr\left[X_{k+1}\ge (t^2)^{2^{k-1}}\right]|X_1\right] \le \E_{X_1}\left[\min\left(1,\frac{X_1^2-1}{t^2-1}\right)\right],
\end{align*}
where the inequality follows from the inductive hypothesis and the fact that all probabilities are at most 1. To bound the above expression, we will first upper bound the argument of $\E_{X_1}[\cdot ]$ in the domain $X_1\ge 1$ by a linear function.

\begin{claim}
\label{clm:upperapprox}
    Let $x\ge 1$ and $t> 1$. We have \[\min\left(1,\frac{x^2-1}{t^2-1}\right)\le \frac{x-1}{t-1}.\]
\end{claim}

\begin{proof}
    If $1 \le x \le t$, then $\frac{x^2-1}{t^2-1}\le 1$, $\frac{x+1}{t+1}\le 1$, and $\frac{x-1}{t-1}\ge 0$. Hence \[\min\left(1,\frac{x^2-1}{t^2-1}\right) = \frac{x^2-1}{t^2-1} = \left(\frac{x-1}{t-1}\right)\left(\frac{x+1}{t+1}\right) \le \frac{x-1}{t-1}.\] If $x > t$, we have $\frac{x-1}{t-1} > 1\ge \min\left(1,\frac{x^2-1}{t^2-1}\right)$.
\end{proof}

With this linear upper bound, we can use the bound on $\E[X_1]$ and linearity of expectation to finish the inductive step as follows.

\begin{align*}\Pr[X_{k+1}\ge t^{2^k}] \le \E_{X_1}\left[\min\left(1,\frac{X_1^2-1}{t^2-1}\right)\right]\le \E_{X_1}\left[\frac{X_1-1}{t-1}\right] \le \frac{X_0^2-1}{t-1}. \end{align*} The desired result follows.
\end{proof}

\subsection{Using the Tail Bound to Analyze Potentials}
With \Cref{lem:potential-tail-bound} in hand, we can now revisit the proof strategy of \Cref{thm:hashingexistence}, and use the tail bound to argue the potentials are minimized with high probability.

\begin{theorem}
\label{thm:hashingwhp}
    Let $\ell\ge 1, u\ge \ell + \log\ell$ be integers, and $n\coloneqq 2^\ell$. Let $S\subset\F_2^{u}$ be of cardinality $n$, and let $\cH$ be the set of all surjective linear maps $h: \F_2^{u}\to\F_2^\ell$.  For any $r\ge 1$,\[\Pr_{h\sim \cH}\left[M(S,h) \ge r\cdot \frac{\log n}{\log\log n}\right] \le \frac{3\log e}{2(r-1)}.\]
\end{theorem}

In the next section, we will establish a stronger tail bound that decays quadratically (\Cref{lem:quad-tail-surj}). \Cref{thm:hashingwhp} suffices for a high constant probability guarantee.

\begin{proof}
Let $k\coloneqq u-\ell$. To pick $h\sim \cH$, we will pick linearly independent vectors $v_1,\dots v_k\in \F_2^u$ uniformly at random, and then pick a uniformly random $h$ with kernel $V\coloneqq \Span(v_{\le k})$.  In particular, for $i=1,\dots, k$ we will iteratively pick $v_{i}\sim \F_2^u\setminus V_{i-1}$. Using $S$ and $v_{\le k}$, define the functions $\{S_i(x)\coloneqq |(x+V_i)\cap S|\}_{0\le i\le k}$, and the potentials $\{\Phi_i(x)\coloneqq \E_{x\sim \F_2^u}[\ell^{S_i(x)}]\}_{i\in[k]}$.
If some preimage of $h$ had load $r\ell/\log \ell$, then $|(x+V)\cap S|\ge r\ell/\log \ell$ for some shift $x$. \Cref{clm:potential-load} then implies \[\Phi_k \ge \frac{\ell^{r\ell/\log\ell}}{2^{u-k}}  = \frac{2^{r\ell}}{2^{\ell}} =  n^{r - 1}.\]

We can easily compute\[\Phi_0 = \E_x[\ell^{S(x)}] = 1 - \frac{n}{2^u} + \frac{n\ell}{2^{u}} < 1 + \frac{\ell}{2^{k}}.\] Consequently, $\Phi_0^2 \le 1 + 2\ell/2^{k} + (\ell/2^k)^2\le 1 + 3\ell/2^k$, where we used the assumption $ { u\ge \ell+\log \ell} \iff \ell/2^k\le 1$. We already know by \Cref{clm:expsquare} that $\E[\Phi_{i+1} | \Phi_{\le i}] \le \Phi_i^2$. Hence, by \Cref{lem:potential-tail-bound} (where all $\Phi_i\ge 1$ since $\ell \ge 1$), we have \begin{align*}\Pr\left[M(S,h)\ge \frac{r\ell}{\log \ell}\right] \le \Pr[\Phi_k \ge n^{r-1}]  &  \le \frac{\Phi_0^2 - 1}{n^{(r-1)/2^{k-1}} - 1} \\ &   \le \frac{3\ell/2^{k}}{(r-1)(\ln n)/2^{k-1}}\tag{$e^x-1\ge x$}  \\ &   = \frac{3\log e}{2(r-1)}.\end{align*}

\end{proof}

\section{Optimal Average Max-Load}
\label{sec:expectation}

In this section, we refine the techniques from \Cref{sec:whp} to prove tail bounds strong enough to establish our main result: optimal expected max-load. The tail bound in \Cref{thm:hashingwhp} decays too slowly to imply this expectation bound. At first glance, strengthening the tail bound appears difficult, as the technical lemma underpinning it, \Cref{lem:potential-tail-bound}, is tight by \Cref{rem:initial-tail-tightness}. Crucially, our potential functions satisfy a \emph{strong} monotonicity property that \Cref{lem:potential-tail-bound} does not exploit: $X_{i+1}-1 \ge 2(X_i - 1)$ for all $i \ge 1$ (\Cref{prop:strides}). 
%For instance, they form an increasing sequence, i.e., $X_{i+1} \ge X_i$. However, this alone is insufficient, because one can construct sequences satisfying $X_{i+1} \ge X_i$ that match the original tail bound asymptotically.
Under this assumption, we can prove a quadratically stronger version of \Cref{lem:potential-tail-bound} (see \Cref{lem:parsed-strong-tail}). 

Interestingly, the conditions required for quadratically decaying appear to be quite delicate for two reasons. \begin{itemize}
 \item A stronger version of \Cref{lem:potential-tail-bound} is not possible under the weaker and more standard monotonicity property $X_{i+1}\ge X_i$. There exists a random sequence $(X_i)$ with $\E[X_{i+1}|X_{\le i}]\le X_i^2$ and $X_{i+1}\ge X_i$ that asymptotically saturates the bound of \Cref{lem:potential-tail-bound}. 
    \item The tight example from \Cref{rem:initial-tail-tightness} satisfies $X_{i+1}-1 \ge 2(X_i - 1)$ for all $i \ge 1$, but not for $i = 0$---showing that even a single exceptional timestep can eliminate any asymptotic improvement on the tails.
   
\end{itemize}

\subsection{A Stronger Tail Bound}
The main result we will show is the following.

\begin{theorem}
\label{lem:parsed-strong-tail}
    Let $X_0 > 1$ be a constant, and let $X_1,\dots, X_k$ be random variables satisfying $X_{i+1}-1\ge 2(X_i-1)$ and $\E[X_{i+1}|X_{\le i}]\le X_i^2$ for all $i$. For $ 1 + 4(X_0-1)\le t\le 2$, we have \[\Pr[X_k\ge t^{2^{k-1}}]\le 48\left(\frac{X_0-1}{t-1}\right)^2.\]
\end{theorem}

We would like to prove the above statement using induction, similar to the proof of \Cref{lem:potential-tail-bound}. However, to make the induction go through, we need to strengthen our inductive hypothesis. For each $i\ge 0$, define the functions
%\[\beta_i(t,\delta) = \frac{4^i\delta^2}{t^{2^i}-1-2^{i+1}\delta}.\] 
\[\beta_i(t,\delta) = \frac{t^{2^i} - (1+2^i\delta)^2}{t^{2^i}-(1+2^{i+1}\delta)} .\] The following stronger claim will be easier to prove via induction.

\begin{lemma}
\label{claim:op-tail-lemma}
    Let $\delta >0$ and $X_0 > 1$ be constants, and let $X_1,\dots, X_k$ be random variables such that $\forall i\ge 0$, $X_i\ge 1 +2^i\delta$, and $\E[X_{i+1}|X_{\le i}]\le X_i^2$. For any $t > 1+2\delta$,  \[\Pr[X_k\ge t^{2^{k-1}}]\le 1-\frac{t-X_0^2}{t-1-2\delta}\prod_{i=1}^{k-1}\beta_i(t,\delta).\]

   % Let $X_0\ge Y_0>1$ be constants, and let $Y_1,\dots, Y_k>1$ be random variables such that $\forall i$, $Y_{i+1}\ge 1 + 2^i(X_i-1)$ almost surely and $\E[Y_{i+1}|Y_{\le i}]\le Y_i^2$. $Y_0^2\ge 2X_0-1$. Define \[\beta_i(t,x) =  \frac{2^{2i}(x-1)^2}{(t^{2^i} - 1) - 2^{i+1}(x -1)} \] Then  \[\Pr[Y_k\ge t^{2^{k-1}}]\le 1-\frac{t-Y_0^2}{t-2X_0+1}\prod_{i=1}^{k-1}(1-\beta_i(t,X_0))\]
\end{lemma}
\begin{remark}
    \Cref{claim:op-tail-lemma} is tight for the sequence \begin{itemize}
        \item $X_0 = 1+\delta$, and $X_1 = \begin{cases}1+2\delta & \text{with probability }\frac{t-X_0^2}{t-1-2\delta}\\ t & \text{otherwise}\end{cases}$,
        \item for $i\ge 1$,
        \begin{itemize}
            \item if $X_i = t^{2^{i-1}}$, set $X_{i+1} = t^{2^i}$
            \item if $X_i = 1 + 2^i\delta$, set $X_{i+1} =\begin{cases} 1+2^{i+1}\delta & \text{with probability }\beta_i(t,\delta) \\ t^{2^i} & \text{otherwise}\end{cases}$.
        \end{itemize}
    \end{itemize}
\end{remark}

\begin{remark}
    Setting $\delta = 0$ recovers \Cref{lem:potential-tail-bound}.
\end{remark}

%\begin{remark}
 %   The motivation behind the $\beta_i$ functions is that $\beta_1(\delta, t)$ will relate to the slope of the best linear function upper bound, akin to \Cref{clm:upperapprox} (see \Cref{claim:mintolingeneral}). In the proof of \Cref{claim:op-tail-lemma}, we will apply the inductive hypothesis with the change of variables  $\delta\gets 2\delta$ and $t\gets t^2$. This motivates defining $\beta_{i+1}(t,\delta) = \beta_{i}(t^2,2\delta) $ (see \Cref{claim:recursebeta}).
%\end{remark}

Before we start the proof, we will establish some preliminary properties of $\beta_i(t,\delta)$.

\begin{claim}
\label{claim:recursebeta}
    $\beta_i(t^2,2\delta) = \beta_{i+1}(t,\delta)$
\end{claim}

\begin{proof}
    \begin{align*}\beta_i(t^2,2\delta) = \frac{(t^2)^{2^i} - (1 + 2^i(2\delta))^2}{(t^2)^{2^i} -(1 + 2^{i+1}(2\delta))}  = \frac{t^{2^{i+1}} - (1+2^{i+1}\delta)^2}{t^{2^{i+1}}-(1+2^{i+2}\delta)} = \beta_{i+1}(t,\delta).\end{align*}
\end{proof}

\begin{claim}
\label{claim:rangebeta}
   For all $i\ge 1$ and $t > 1+2\delta$, $0\le \beta_i(t,\delta)\le 1.$
\end{claim}

\begin{proof}
For the lower bound, note $t^{2^i}\ge 1+2^{i+1}\delta$ for all $i$ (by \Cref{claim:silly}). Consequently, the numerator and denominator of $\beta_i(t,\delta)$ are both positive.
For the upper bound, we see ${1- \beta_i(t,\delta)} = \frac{4^i\delta^2}{t^{2^i}-(1+2^{i+1}\delta)}$. The numerator is trivially nonnegative, and the denominator is positive by \Cref{claim:silly}.    
\end{proof}
%We now show a generalization of \Cref{clm:upperapprox}, which will lower bound an expression we will encounter by a linear function.

With these tools, we are ready to prove \Cref{claim:op-tail-lemma}.

\begin{proof}[Proof of \Cref{claim:op-tail-lemma}]
We apply induction on $k$. For the base case $k=1$, we apply Markov's inequality on the nonnegative random variable $X_1-1-2\delta$ to yield \begin{align*}\Pr[X_1\ge t]  = \Pr[X_1-1-2\delta\ge t-1-2\delta]   \le \frac{X_0^2-1-2\delta}{t-1 - 2\delta}   = 1-\frac{t-X_0^2}{t-1-2\delta}.\end{align*}  Note the implicit use of $t> 1+2\delta$ in the application of Markov's inequality.

Now assume the lemma for $k$. We will now prove it for $k+1$. Write the tail probability as
\begin{align}\Pr[X_{k+1}\ge t^{2^k}] = \E_{X_1}[\Pr[X_{k+1}\ge (t^2)^{2^{k-1}} | X_1]].\label{eqn:induct1}\end{align} For fixed $X_1$, it follows that $X_2,\dots, X_{k+1}$ is a sequence of random variables of length $k$ satisfying $X_{i+1}\ge 1 + 2^i(2\delta)$ and $\E[X_{i+2}|X_1,X_2,\dots, X_{i+1}] \le X_{i+1}^2$ for all $i\in [k]$. Furthermore, $t^2\ge 1+2(2\delta)$ by \Cref{claim:silly}. Hence, all assumptions of the inductive hypothesis are satisfied with the instantiation $\delta\gets 2\delta$, $t\gets t^2$. Utilizing this as well as the fact all probabilities are bounded by 1, we obtain \begin{align} \E_{X_1}[\Pr[X_{k+1}\ge (t^2)^{2^{k-1}} | X_1]]\notag  &\le \E_{X_1}\left[\min \left(1, 1-\frac{t^2-X_1^2}{t^2-1-4\delta}\prod_{i=1}^{k-1}\beta_i(t^2, 2\delta)\right)\right]\notag  \\ & = 1 - \E_{X_1}\left[\max\left(0,\frac{t^2-X_1^2}{t^2-1-4\delta}\prod_{i=2}^k \beta_i(t,\delta)\right)\right] \notag \\ & = 1 - \E_{X_1}\left[\max\left(0, \frac{t^2-X_1^2}{t^2-1-4\delta}\right)\right] \prod_{i=2}^k\beta_i(t,\delta) \label{eqn:induct2}\end{align}
where the last equality used the fact $\prod_{i=2}^k\beta_i(t,\delta)\ge 0$. To bound this expression, we would like to mimic the intuition of \Cref{clm:upperapprox} and bound the argument of $\E_{X_1}[\cdot]$ by a linear function in the domain $X_1\ge 1+2\delta$. The following claim does so. \begin{claim}
\label{claim:mintolingeneral}
    Let $\delta \ge 0$, $x\ge 1+2\delta$, and $t > 1+2\delta$. Then \[\max\left(0,\frac{t^2-x^2}{t^2-1-4\delta}\right)\ge \left(\frac{t-x}{t-1-2\delta}\right)\beta_1(t,\delta). \]
\end{claim}
Assuming this claim is true for now (see \Cref{rem:geointuit} for geometric intuition), we can use the fact $X_1\ge 1+2\delta$, $\beta_i(t,\delta)\ge 0$,  \Cref{claim:mintolingeneral}, and linearity of expectation to bound \begin{align}1 - \E_{X_1} \left[\max\left(0,\frac{t^2-X_1^2}{t^2-1-4\delta}\right)\right] \prod_{i=2}^k\beta_i(t,\delta)\notag  & \le 1 - \E_{X_1}\left[\frac{t-X_1}{t-1-2\delta}\cdot \beta_1(t,\delta)\right]\prod_{i=2}^k \beta_i(t,\delta)\notag \\ & \le 1 - \frac{t-X_0^2}{t-1-2\delta}\prod_{i=1}^k\beta_i(t,\delta).\label{eqn:induct3} \end{align} Combining \eqref{eqn:induct1},\eqref{eqn:induct2}, and \eqref{eqn:induct3} gives us \[\Pr[X_{k+1}\ge t^{2^k}]\le 1 - \frac{t-X_0^2}{t-1-2\delta}\prod_{i=1}^k \beta_i(t,\delta).\] The desired result follows by induction.
\end{proof}

We now prove \Cref{claim:mintolingeneral}. Notice when $\delta = 0$ this is exactly \Cref{clm:upperapprox}.
\begin{proof}[Proof of \Cref{claim:mintolingeneral}]
    If $x  >  t$, then  $\frac{t-x}{t-1-2\delta} < 0$ as $t > 1+2\delta$. Since we know $\beta_i(t,\delta)\ge 0$ by \Cref{claim:rangebeta}, it follows\[\left(\frac{t-x}{t-1-2\delta}\right)\beta_1(t,\delta) < 0 \le \max\left(0, \frac{t^2-x^2}{t^2-1-4\delta}\right).\] If $x \le t$, then  $\frac{t-x}{t-1-2\delta}\ge 0$ and $\frac{t^2-x^2}{t^2-1-4\delta}\ge 0$, since $t > 1+2\delta$ and $t^2 > 1+4\delta$ (by \Cref{claim:silly}). Hence, \begin{align*}\max\left(0, \frac{t^2-x^2}{t^2-1-4\delta}\right)  = \frac{t^2-x^2}{t^2-1-4\delta}   &  = \left(\frac{t-x}{t-1-2\delta}\right)\left(\frac{(t+x)(t-1-2\delta)}{t^2-1-4\delta}\right) \\ &    \ge  \left(\frac{t-x}{t-1-2\delta}\right)\frac{(t+1+2\delta)(t-1-2\delta)}{t^2-1-4\delta}\\ & = \left(\frac{t-x}{t-1-2\delta}\right)\left(\frac{t^2-(1+2\delta)^2}{t^2-1-4\delta}\right) \\ & = \left(\frac{t-x}{t-1-2\delta}\right)\beta_1(t,\delta).\end{align*} 
    %The conclusion now follows, as \begin{align*}\min\left(1,1-\frac{t^2-x^2}{t^2-1-4\delta}\alpha\right) & = 1 - \max(0,\frac{t^2-x^2}{t^2-1-4\delta}\alpha) \\ &  = 1 - \alpha\max(0,\frac{t^2-x^2}{t^2-1-4\delta})\\ & \le 1-\frac{t-x}{t-1-2\delta}(1-\beta_1(t,\delta))\alpha\end{align*}
\end{proof}

\begin{remark}\label{rem:geointuit} As per the diagram below, the LHS of \Cref{claim:mintolingeneral} is a downward-facing parabola which flattens out to $0$ for $x\ge t$. By convexity, a lower bound for $x\ge 1+2\delta$ will be the line that intersects the parabola at $x=1+2\delta$ and $x=t$. The equation of this line is exactly the RHS.      
    
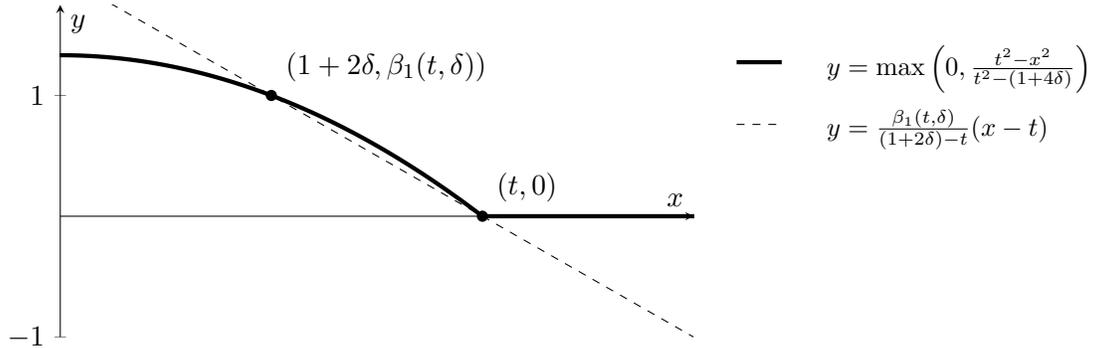
\begin{figure}[h!]

\centering
\begin{tikzpicture}

    \begin{axis}[
        axis lines=middle,
        width=10cm,
        height=6cm,
        xtick distance = 4,
        ytick distance = 1,
        xlabel={$x$},
        ylabel={$y$},
        domain=-1:4,
        samples=200,
        legend cell align={left},
        legend style={
            at={(1.65,.56)}, % Position in bottom right corner
            anchor=south east, % Anchor the legend to the bottom right
            row sep=0.25cm, % Increase vertical space between legend entries
            column sep=0.5cm, % Increase horizontal space if multiple columns are used
            draw=none, % Remove border around the legend
            font=\small % Reduce the size of the legend
        },
        xmin=0, xmax=3,
        ymin=-1, ymax=1.75
    ]

    \addplot[black, ultra thick] {max(0, (4 - x^2)/3)};
    \addlegendentry{$y = \max\left(0,\frac{t^2-x^2}{t^2-(1+4\delta)}\right)$ }

    \addplot[black, dashed, domain=-1:3] {(2-x)};
    \addlegendentry{$y = \frac{\beta_1(t,\delta)}{(1+2\delta) - t}(x-t)$ }

    \node[label={45:{$(t,0)$}},circle,fill,inner sep=1.5pt] at (axis cs:{2},0) {};

    \node[label={45:{$(1+2\delta,\beta_1(t,\delta))$}},circle,fill,inner sep=1.5pt] at (axis cs:{1},1) {};

    \end{axis}
    
\end{tikzpicture}
\label{fig:geointuition}
\end{figure}
\end{remark}

At this point, we have unconditionally proven \Cref{claim:op-tail-lemma}. Using \Cref{claim:op-tail-lemma}, we can now show \Cref{lem:parsed-strong-tail} is true.

\begin{proof}[Proof of \Cref{lem:parsed-strong-tail}]
    Let $\delta \coloneqq X_0-1$. Notice that by composing the assumed inequality, the random variables satisfy $X_i-1\ge 2^i(X_0-1) = 2^i\delta$ for each $i$. Hence, we can apply \Cref{claim:op-tail-lemma} and use the fact $\beta_0(t,\delta) = \frac{t-(1+\delta)^2}{t-1-2\delta} = \frac{t-X_0^2}{t-1-2\delta}$ to yield \begin{align*}\Pr[X_k\ge t^{2^{k-1}}]  \le 1 - \left(\frac{t-X_0^2}{t-1-2\delta}\right)\prod_{i=1}^k\beta_i(t,\delta)  = 1 - \prod_{i=0}^k\beta_i(t,\delta). \end{align*} Since \Cref{claim:rangebeta} tells us $0\le \beta_i(t,\delta) \leq 1$ for all $i$, we can treat these quantities as probabilities. Consider $k+1$ independent and biased coins, where coin $i$ has probability $\beta_i(t,\delta)$ of showing heads. By the union bound on the event that at least one tail shows, we have
    %Using $0\le \beta_i(t,\delta) \leq 1$ for all $i \geq 1$ from \Cref{claim:rangebeta} and the fact $1-\beta_i(t,\delta) = \frac{4^i\delta^2}{t^{2^i}-(1+2^{i+1}\delta)}$, we can further bound 
    \[ 1 - \prod_{i=0}^k\beta_i(t,\delta)   
    %= \sum_{i=0}^k (1-\beta_i(t,\delta))\prod_{j=i+1}^k\beta_i(t,\delta)   
    \le \sum_{i=0}^k(1-\beta_i(t,\delta))   = \sum_{i=0}^k\frac{4^i\delta^2}{t^{2^i}-1 - 2^{i+1}\delta}   \le 2\delta^2\sum_{i=0}^k \frac{4^i}{t^{2^i}-1},\] where we used the fact that for all $i$, $t^{2^i} - 1\ge 2^i(t-1)\ge  2^{i+2}\delta$ by \Cref{claim:silly} and the theorem assumption. For $i\le \log\left(\frac{5}{t-1}\right)$, we have \begin{align*}\sum_{i\le \log\left(\frac{5}{t-1}\right)}\frac{4^i}{t^{2^i}-1}  =\sum_{i\le \log\left(\frac{5}{t-1}\right)}\frac{2^i}{\ln t}\cdot \frac{2^i\ln t}{e^{2^i\ln t}-1}  & \le \frac{1}{\ln t}\sum_{i\le \log\left(\frac{5}{t-1}\right)} 2^i\tag{$\frac{x}{e^x-1}\le 1$} \\ & \le \frac{1}{\ln t}\cdot \frac{10}{t-1}\\ & \le \frac{t}{t-1}\cdot \frac{10}{t-1} \tag{$1+x\ge e^{\frac{x}{1+x}}$}\\ & \le \frac{20}{(t-1)^2},\end{align*} For $i > \log\left(\frac{5}{t-1}\right)$, we can bound  \begin{align*}\sum_{i> \log\left(\frac{5}{t-1}\right)}\frac{4^i}{t^{2^i}-1}  = \sum_{j\ge 1} \frac{4^{\log(5/(t-1))}4^j}{(t^{5/(t-1)})^{2^j} - 1}  & \le \left(\frac{5}{t-1}\right)^2\sum_{j\ge 1}\frac{4^j}{(1 + \frac{5}{t-1}(t-1))^{2^j} - 1} \tag{\Cref{claim:silly}} \\ & \le \left(\frac{5}{t-1}\right)^2\sum_{j\ge 1}(4/35)^j \\ & \le \frac{4}{(t-1)^2}.  \end{align*} Hence we have \begin{align*}\Pr[X_k\ge t^{2^{k-1}}]  \le 2\delta^2\sum_{i=0}^k \frac{4^i}{t^{2^i}-1} \le 2\delta^2\left(\frac{20}{(t-1)^2}+\frac{4}{(t-1)^2}\right)  &  = 48\left(\frac{X_0-1}{t-1}\right)^2 .\end{align*}
\end{proof}

\subsection{Applying the Tail Bound To Our Potentials}

We will strengthen the tail bounds of \Cref{thm:hashingwhp}, and also generalize \Cref{thm:hashingwhp} to the setting of $m$ balls and $n$ bins, for $m\neq n$. For brevity, define \[\opt(m,n) = \begin{cases} \frac{\log n}{\log\left(\frac{n\ln n}{m}\right)} & m\le \frac{1}2 n\log n \\ \frac{m}n & m > \frac{1}2n\log n \end{cases}\] to be the function that outputs the maximum load obtained when a fully random hash maps $m$ balls to $n$ bins.

\begin{theorem}
\label{lem:quad-tail-surj}
Let $u,\ell,m\ge 1$ be integers such that $u\ge \ell + 2\log(\ell m)$, and let $n\coloneqq 2^\ell$. Let $\cH$ be the set of surjective linear maps $\F_2^u\to\F_2^\ell$. Let $S\subset \F_2^u$ be a subset of size $m$. Then for any real $r\ge 6$, \[\Pr_{h\sim \cH}\left[M(S,h)\ge r\cdot \opt(m,n)\right]\le \frac{48}{(r-2)^2}\]
\end{theorem}

\begin{proof}
     % It suffices to show the claim for large enough $u$, say $u\ge \ell + \log \ell + \log m$. To see why this doesn't lose generality, let $\cH'$ be a random uniform surjective linear map $h:\F_2^{u'}\to\F_2^\ell$ for some $u' < u$. Now consider a $u'$-dimensional subspace $V\subset \F_2^u$. There is an isomorphism $\iota: \F_2^{u'}\to V$. For $h\sim \cH$, $h|_V:V\to \F_2^\ell$ will be a random linear surjective map, and for any $T\subset V$, \begin{align*}M(T,h|_V) & = \max_{y\in \F_2^\ell} |(h|_V)^{-1}(y)\cap S| \\ & = \max_{y\in \F_2^\ell} |h^{-1}(y)\cap V \cap T| \\ & = \max_{y\in \F_2^\ell} |h^{-1}(y)\cap T| \\ & =M(T,h). \end{align*} Consequently, for any $M$,  \[\Pr_{h'\sim \cH}[M(S,h')\ge M] = \Pr_{h\sim \cH}[M(\iota(S), h|_V)\ge M]= \Pr_{h\sim \cH}[M(\iota(S), h)\ge M]. \]
    
    Define $k\coloneqq u-\ell$. We assume $r\le m$, since $M(S,h)\le |S| = m$ for any $h$. We will pick $h\sim\cH$ by iteratively sampling $v_{i+1}\sim \F_2^u\setminus V_i$, and then picking random $h$ with $\ker(h) = V$. Define the associated functions $\{S_i(x)\coloneqq |(x+V_i)\cap S|\}_{0\le i\le k}$. We will split into two cases.
    
    \noindent \textbf{Case 1: $m\le \frac{1}2 n\log n$}. 
    
    Set $b=n\ell/m$ and $\Phi_i \coloneqq \E_{x\sim \F_2^u}[b^{S_i(x)}]$ for $0\le i\le k$. Then \[\Phi_0 = \left(1-\frac{m}{2^u}\right)\cdot 1 + \frac{m}{2^u}\cdot b\le 1 + \frac{bm}{2^u}= 1 + \frac{\ell}{2^k}.\] By \Cref{clm:load-potential} and the fact $\opt(m,n) = \ell/\log b$ in this case, we have \begin{align*}\Pr[M(S,h) \ge r \cdot \opt(m,n)]  \le \Pr\left[\Phi_k\ge \frac{b^{r\cdot \opt(m,n)}}{n}\right]  = \Pr[\Phi_k\ge n^{r-1}].\end{align*} Note that for $t = 1 + r\ell/2^k$, $t^{2^{k-1}} \le e^{r\ell/2}  < n^{r-1}$. Furthermore, as $4\le r\le m$ and $u\ge \ell + \log( \ell m)$, we have $t-1\ge 4\ell/2^k\ge  4(\Phi_0-1)$ and $t-1 \le  mn\ell/2^u\le 1$. Finally, we have $\Phi_{i+1}-1\ge 2(\Phi_i-1)$ for all $i$ by \Cref{prop:strides}. Hence, we can use \Cref{lem:parsed-strong-tail} to bound \begin{align*}\Pr[\Phi_k\ge n^{r-1}]  \le \Pr[\Phi_k\ge t^{2^{k-1}}]   \le 48\left(\frac{\Phi_0-1}{t-1}\right)^2  \le 48\left(\frac{\ell/2^k}{r\ell/2^k}\right)^2 =\frac{48}{r^2}.  \end{align*} 
    
\noindent \textbf{Case 2:  $m > \frac{1}2n\log n$}. 

Define $\Phi_i \coloneqq \E_{x\sim\F_2^u}[2^{S_i(x)}]$ for $0\le i\le k$. Then \[\Phi_0 = 1-\frac{m}{2^u} + \frac{2m}{2^u} = 1 + \frac{m}{2^u}.\] By \Cref{clm:load-potential} and the fact $\opt(m,n) = m/n$ in this regime, \begin{align*}\Pr[M(S,h)\ge r\cdot \opt(m,n)] & = \Pr[M(S,h)\ge rm/n] \\ & \le \Pr[\Phi_k\ge 2^{rm/n-\ell}] \\ &  \le \Pr[\Phi_k\ge 2^{(r-2)m/n}]\tag{$m\ge \frac{1}2n\ell$ }\end{align*} For $t \coloneqq 1 + (r-2)m/2^u$, we have $t^{2^{k-1}}\le e^{(r-2)m/2n}\le 2^{(r-2)m/n}$. Since $6\le r\le m$ and $u\ge \ell + 2\log m$, we can deduce $t-1\ge 4m/2^u \ge 4(\Phi_0-1) $ and $t-1\le m^2/2^u\le 1$. Furthermore, we have $\Phi_{i+1}-1\ge 2(\Phi_i-1)$ for all $i$ by \Cref{prop:strides}. Hence, by \Cref{claim:op-tail-lemma}, we have \begin{align*}\Pr[\Phi_k\ge t^{2^{k-1}}]  \le 48\left(\frac{\Phi_0-1}{t-1}\right)^2   \le 48\left(\frac{m/2^u}{(r-2)m/2^u}\right)^2   = \frac{48}{(r-2)^2}.\end{align*} 

\end{proof}

With a simple argument (whose proof is deferred to \Cref{app:removing-surjectivity}), we can remove the artificial lower-bound condition on $u$ and the surjectivity condition on $h$. 

\begin{theorem}
\label{thm:quadratic-tail-final}
    Let $u\ge \ell\ge 1$ and $m\le 2^u$ be integers.  Let $n\coloneqq 2^\ell$. Let $h:\F_2^u\to\F_2^\ell$ be a random linear map. For any $S\subset \F_2^u$ of size $m$ and $r\ge 6$, we have \[\Pr_{h\sim \cH} [M(S,h)\ge r\cdot \opt(m,n)]\le \frac{49}{(r-2)^2}.\]
\end{theorem}
With these strong tails, optimal expected max-load is a simple corollary.

\begin{theorem}
    Let $u\ge \ell\ge 1$ be integers, and $n\coloneqq 2^\ell$. For uniformly random linear map $h:\F_2^u\to \F_2^\ell$, \[\E_h[M(S,h)]\le 16\cdot \opt(m,n).\]
\end{theorem}

\begin{proof}
    \begin{align*}\E_h[M(S,h)] & = \int_{0}^\infty\Pr[M(S,h)\ge t]dt \\ & \le \int_0^{9\cdot \opt(m,n)} 1dt + \int_{9\cdot \opt(m,n)}^\infty\Pr[M(S,h)\ge t]dt\\ & \le 9\opt(m,n)+ \opt(m,n)\int_{9}^\infty \Pr[M(S,h)\ge r\opt(m,n)]dr  \\ & \le 9\opt(m,n) + 49\opt(m,n)\int_{9}^\infty\frac{1}{(r-2)^2}dr \tag{\Cref{thm:quadratic-tail-final}}\\ & =   16\cdot \opt(m,n). \end{align*}
\end{proof}

\section{Two-Sided Bounds}
\label{sec:two-sided}

%We first confirm that all of the max-load upper bounds from \Cref{thm:hashingwhp} are tight in a very strong manner. We first start with the regime $m = O(n\log n)$. In this setting, most ball sets $S$ will not have a \emph{single} linear map with $M(S,h) \le (1-o(1))\opt(m,n)$.

%\begin{theorem}
%    Let $m,n\in \N$ with $2^u = m^2$. There exists a constant $C$ and a subset $S\subset\F_2^u$ of cardinality at most $m$ such that \emph{every} linear map $h:\F_2^{u}\to\F_2^\ell$ has $M(S,h) > (1-o(1)) \opt(m,n)$.
%\end{theorem}

%The proof uses standard techniques from the classical balls-and-bins problem, so we defer it to the appendix.

In the regime of $m = \Omega(n\log n)$, we can give two-sided bounds on all bins. In particular, we can show for any set of $m$ balls, a random linear map will hash $\Theta(m/n)$ balls to each bin with high probability.

\begin{theorem}
\label{thm:twosided}
    Let $0 <\eps < 1/2$ be a constant. There exists constants $C_1 < 1 < C_2$ depending on $\eps$ such that for $m\ge C_1^{-1} n\log n$ and any $S\subset\F_2^u$ of cardinality $m$, a uniformly random linear map $h:\F_2^u\to \F_2^\ell$ satisfies \[\Pr\left[\forall y\in \F_2^\ell, \frac{C_1 m}n\le |h^{-1}(y)\cap S|\le \frac{C_2 m}n\right] \ge 1-\eps.\] Furthermore, $C_1 = \Omega(\eps^{74})$ and $C_2 = O(\eps^{-1/2})$.
\end{theorem}

We reiterate that the condition $m=\Omega(n\log n)$ is necessary: there exists $S$ of size at least $ 0.69n \log n$ such that \emph{every} linear map has at least one empty bin (\cite{alon99}, Proposition 2.2). Interestingly, this two-sided bound is also proven using potential functions. To prove this, we will require the following tail bound.

\begin{lemma}
\label{lem:square-decay}
    Let $0<\eps < 1$ and $0 \le X_0 < 1$ be fixed, and let $ 1 > X_1\ge X_2\ge \dots  > 0$ be random variables satisfying $X_{i+1}\ge 2X_i - 1$ and $\E[X_{i+1}|X_{\le i}]\le X_i^2$. Then for ${C_\eps = (1-X_0)^{25}(\eps/2)^{50}\log(1/\eps)}$ and any $s$,

    \[\Pr\left[X_s\ge 2^{-C_\eps 2^s}\right]\le \eps.\]

\end{lemma}

We will first prove \Cref{thm:twosided} assuming this tail bound, and then prove the tail bound afterwards. 

\begin{proof}[Proof of \Cref{thm:twosided}]
   
    We first focus on the lower bound. Set $t = \log(4m/\eps)$ and factor $h = h_1\circ h_2$, where $h_1:\F_2^u\to\F_2^t$ is a uniformly random linear map, and $h_2:\F_2^t\to\F_2^\ell$ is a uniformly random \emph{surjective} linear map. $h_1$ will collide any fixed pair of elements in $S$ with probability $1/2^t$. Hence the expected number of pairwise collisions among $S$ is $\binom{m}2 \frac{1}{2^t}\le \frac{m^2}{2^{t+1}}$. Let $\cE$ denote the event in which there are at least $m/2$ pairs that collide. By Markov's inequality, \[\Pr[\cE] \le \frac{m^2}{2^{t+1}}\cdot \frac{2}{m} = \frac{m}{2^t} = \frac{\eps}4 .\]
    
    Letting $k\coloneqq t-\ell$, we will sample random linearly independent vectors $v_1,\dots, v_k\in \F_2^t$, and consider $h_2$ with kernel $V\coloneqq \Span(v_{\le k})$. Using the set $h_1(S)$ and vectors $v_{\le k}$, we construct the associated functions $\{S_i(x)\coloneqq |(x+V_i)\cap h_1(S)|\}_{0\le i\le k}$ and potentials $\{\Phi_i \coloneqq \E_{x\sim \F_2^u}[(1/2)^{S_i(x)}]\}_{0\le i\le k}$.

    Conditioning on $\neg\cE$, we have $|h_1(S)|\ge m-m/2 =  m/2$, and so the density of $h_1(S)\subset \F_2^t$ is lower bounded by $\frac{m/2}{4m/\eps} = \eps/8$. Consequently, we have \[\Phi_0 \le  \left(1-\frac{\eps}8\right)\cdot 1 + \frac{\eps}{8}\cdot \frac{1}2 = 1 - \frac{\eps}{16}.\]
   \Cref{clm:expsquare} and \Cref{prop:strides} show that $(\Phi_i)$ satisfy the premise of \Cref{lem:square-decay}. Hence, we can use \Cref{lem:square-decay} and the fact $\ell \le C_1m/n$ to deduce 
    \begin{align*}\Pr[\Phi_k > 2^{-C_1m/n - \ell}]   \le \Pr[\Phi_k > 2^{-2C_1m/n}]   = \Pr[\Phi_k > 2^{- (\eps/2)C_1 2^k}]  \le \eps/4\end{align*} for $C_1 \coloneqq \frac{2}\eps(1-\Phi_0)^{25}(\eps/8)^{50}\log(4/\eps) = \Omega(\eps^{74})$. 
    Now note for all $y\in \F_2^\ell$, we have \[h^{-1}(y) = \bigsqcup_{z\in h_2^{-1}(y)}h_1^{-1}(z).\] Hence, \begin{align*}|h^{-1}(y)\cap S|  = \sum_{z\in h_2^{-1}(y)}|h_1^{-1}(z)\cap S|  \ge \sum_{z\in h_2^{-1}(y)}\mathbf{1}(z\in h_1(S)) = |h_2^{-1}(y)\cap h_1(S)|.  \end{align*}
    Therefore, it follows \begin{align}\Pr\left[\forall y  \in \F_2^\ell, |h^{-1}(y)\cap S| < C_1 \frac{m}n\right]
     & \le \Pr\left[\forall y\in \F_2^\ell,|h_2^{-1}(y)\cap h_1(S)| < C_1\frac{m}n\right] \nonumber\\ & \le \frac{\eps}4 + \Pr\left[\forall y\in \F_2^\ell,|h_2^{-1}(y)\cap h_1(S)| < C_1\frac{m}n \big| \neg \cE\right] \nonumber
    \\ & \le \frac{\eps}4 + \Pr[\Phi_k > 2^{-C_1m/n - \ell} | \neg \cE]\nonumber
    \tag{\Cref{clm:blanketing}}
    \\ & \le \frac{\eps}4 + \frac{\eps}4 = 
    \frac{\eps}2.\label{eqn:upperbound}\end{align}
We already have from \Cref{thm:quadratic-tail-final} that there exists $C_2=7\eps^{-1/2}+2$ such that ${\Pr[M(S,h) > C_2m/n]}\le \eps/2$. Taking a union bound over this and \Cref{eqn:upperbound} gives the desired result.

\end{proof}

\begin{proof}[Proof of \Cref{lem:square-decay}]
The proof will resemble that of Theorem 7 in \cite{alon99}. Call $i$ a \emph{stride} if $X_{i-1} > 1/2$ and $(1-X_i)\ge \frac{5}4(1-X_{i-1})$, or if $X_{i-1} \le 1/2$ and $X_i\le \frac{3}4X_{i-1}$. We will show on any conditioning of $X_{\le i}$, $i+1$ is a stride with probability $\ge 1/3$. We split into two cases.

\begin{itemize}
\item If $X_{i}\le 1/2$, then $\E[X_{i+1}]\le X_{i}^2\le \frac{1}2X_i$. Therefore, by Markov's inequality, $\Pr[X_{i+1} > \frac{3}4X_i] < \frac{X_i/2}{3X_i/4} = 2/3$, so $i+1$ is a stride with probability at least $1/3$.

\item If $X_i > 1/2$, we will instead apply Markov's inequality on $1-2X_i+X_{i+1}$. Note \[\E[1-2X_i + X_{i+1} | X_{\le i}] = 1-2X_i + X_i^2 = (1-X_i)^2.\] Hence, conditioned on $X_{\le i}$, we have \begin{align*}\Pr_{X_{i+1}}\left[1-X_{i+1}< \frac{5}4(1-X_i)\right] & = \Pr_{X_{i+1}}\left [1-2X_i + X_{i+1} > \frac{3}4(1-X_i)]\right]   < \frac{(1-X_i)^2}{3(1-X_i)/4} < \frac{2}3,\end{align*} so $i+1$ is a stride with probability $\ge 1/3$ in this case as well. 

\end{itemize}

Let $j$ be the first integer such that $X_{j+1} < 1/2$, and let $k$ be the first integer such that $X_{k+1} < \eps^2/8$. Let $s_1$ be the number of strides in $[j]$, and let $s_2$ be the number of strides in $\{j+1,\dots, k\}$. We observe that $(1-X_0)(\frac{5}{4})^{s_1}\le \frac{1}2$, and so $s_1\le \log_{5/4}(1/2(1-X_0))$. Similarly, we must have $\frac{1}2\left(\frac{3}4\right)^{s_2}\ge \frac{\eps^2}8$, implying that $s_2\le \log_{4/3}(4/\eps^2)$. Therefore, there must be at most $s_1+s_2 \le \log_{5/4}(4/(1-X_0)\eps^2)$ strides in $[k]$. Define $k^*\coloneqq 8\log_{5/4}(4/(1-X_0)\eps^2)$. Let  $f(X_0, X_1,\dots, X_{k^*})$ evaluate the number of strides in the first $k^*$ steps, and define $Y_i \coloneqq \E[f(X_{\le k^*}) | X_{ < i}]$. Notice $(Y_i)$ is a Doob martingale with $|Y_{i+1}-Y_i|\le 1$ for all $i$. Furthermore, since we showed each index is a stride with probability $\ge 1/3$,  \[Y_0 \ge k^*/3 = (8/3)\log_{5/4}(4/(1-X_0)\eps^2) \ge 16\log(2/\eps)\] Hence, by Azuma's inequality (\Cref{fact:azuma}) and the above string of inequalities, \begin{align*}\Pr[k\ge k^*] &\le \Pr\left[f(X_{\le k^*})  \le  \log_{5/4}\left(\frac{4}{(1-X_0)\eps^2}\right)\right] \\ & \le \Pr\big[Y_{k^*} \le (3/8)Y_0\big] \\ &   \le e^{-\frac{(5Y_0/8)^2}{2k^*}}  \le e^{-\frac{(5/8)^2}6( 16\log (2/\eps))}  \le \eps/2.\end{align*} 
If $s \le k^*$, we have just shown with probability at most $\eps/2$ that $X_s\ge X_{k}\ge \eps^2/8\ge 2^{-C_\eps 2^s}$ for any $C_\eps\le \log(8/\eps^2)$. Henceforth, we will assume $s > k^*$. Conditioned on $k < k^*$, we have $X_{k^*} < \eps^2/8$. Let $\cE_i$ be the event $X_{k^*+i}< \eps^{2^{i} + 1}/2^{i+3}$ and denote $\cE_{\le t} = \bigwedge_{i=1}^t\cE_i$. Notice that $\E[X_{k^*+i+1}|\cE_{\le i}]\le \eps^{2^{i+1}+2}/2^{2i+6}$, so by Markov's inequality, \[\Pr[\neg \cE_{i+1}| \cE_{\le i}]\le \frac{\eps^{2^{i+1}+2}/2^{2i+6}}{\eps^{2^{i+1}+1}/2^{i+4}} \le \eps/2^{i+2}.\]  Hence for any $t$ we have \[\Pr\left[\neg\cE_{\le t}\right] = \sum_{i=0}^t \Pr\left[\neg \cE_{i+1}\wedge \cE_{\le i}\right]\le \sum_{i=0}^{t} \Pr[\neg\cE_{i+1}|\cE_{\le i}]\le \sum_{i\ge 0}\eps/2^{i+2}\le \eps/2.\] 

Consequently by a union bound, $k < k^*$ and $\cE_{\le (s-k^*)}$ occurs with probability at least $1-\eps$. In the event this happens, and noting $k^*=8\log_{5/4}(4/(1-X_0)\eps^2)\le 25\log(4/(1-X_0)\eps^2)$, it follows \[X_s \le  \frac{\eps^{2^{s-k^*}+1}}{2^{s-k^*+3}}\le \eps^{2^s((1-X_0)\eps^2/4)^{25} + 1} \le 2^{-C_\eps 2^s}\] for $C_\eps = (1-X_0)^{25}(\eps/2)^{50}\log(1/\eps)$. The desired result follows. 

\end{proof}

\section*{Acknowledgements}
    We thank Raghu Meka for discussions, one of which led to this question. We also thank Jesse Goodman for discussions, as well as Jeffrey Champion, Sabee Grewal, and Lin Lin Lee for comments that improved our presentation.  

\bibliographystyle{alphaurl} \bibliography{references}

\newcommand{\etalchar}[1]{$^{#1}$}
\begin{thebibliography}{ADM{\etalchar{+}}97}

\bibitem[ADM{\etalchar{+}}97]{alon99}
Noga Alon, Martin Dietzfelbinger, Peter~Bro Miltersen, Erez Petrank, and G\'{a}bor Tardos.
\newblock Is linear hashing good?
\newblock In {\em Proceedings of the Twenty-Ninth Annual ACM Symposium on Theory of Computing (STOC)}, page 465–474, 1997.
\newblock \href {https://doi.org/10.1145/258533.258639} {\path{doi:10.1145/258533.258639}}.

\bibitem[Bab18]{babka18}
Martin Babka.
\newblock A note on the size of largest bins using placement with linear transformations.
\newblock 2018.
\newblock \href {https://arxiv.org/abs/1810.04161} {\path{arXiv:1810.04161}}.

\bibitem[BGG94]{BGG94}
Mihir Bellare, Oded Goldreich, and Shafi Goldwasser.
\newblock Incremental cryptography: The case of hashing and signing.
\newblock In Yvo~G. Desmedt, editor, {\em Advances in Cryptology — CRYPTO 1994}, volume 839 of {\em Lecture Notes in Computer Science}, pages 216--233, 1994.
\newblock \href {https://doi.org/10.1007/3-540-48658-5_22} {\path{doi:10.1007/3-540-48658-5_22}}.

\bibitem[BGG95]{BGG95}
Mihir Bellare, Oded Goldreich, and Shafi Goldwasser.
\newblock Incremental cryptography and application to virus protection.
\newblock In {\em Proceedings of the Twenty-Seventh Annual ACM Symposium on Theory of Computing (STOC)}, page 45–56, 1995.
\newblock \href {https://doi.org/10.1145/225058.225080} {\path{doi:10.1145/225058.225080}}.

\bibitem[CMV13]{cmv13-simplehashingonentropicsources}
Kai-Min Chung, Michael Mitzenmacher, and Salil Vadhan.
\newblock Why simple hash functions work: Exploiting the entropy in a data stream.
\newblock {\em Theory of Computing}, 9(30):897--945, 2013.
\newblock \href {https://doi.org/10.4086/toc.2013.v009a030} {\path{doi:10.4086/toc.2013.v009a030}}.

\bibitem[CRSW13]{crsw13}
L.~Elisa Celis, Omer Reingold, Gil Segev, and Udi Wieder.
\newblock Balls and bins: Smaller hash families and faster evaluation.
\newblock {\em SIAM Journal on Computing}, 42(3):1030--1050, 2013.
\newblock \href {https://doi.org/10.1137/120871626} {\path{doi:10.1137/120871626}}.

\bibitem[CW79]{cw79}
J.Lawrence Carter and Mark~N. Wegman.
\newblock Universal classes of hash functions.
\newblock {\em Journal of Computer and System Sciences}, 18(2):143--154, 1979.
\newblock \href {https://doi.org/10.1016/0022-0000(79)90044-8} {\path{doi:10.1016/0022-0000(79)90044-8}}.

\bibitem[DD22]{dhardvir}
Manik Dhar and Zeev Dvir.
\newblock Linear hashing with $\ell_\infty$ guarantees and two-sided kakeya bounds.
\newblock In {\em 2022 IEEE 63rd Annual Symposium on Foundations of Computer Science (FOCS)}, pages 419--428, 2022.
\newblock \href {https://doi.org/10.1109/FOCS54457.2022.00047} {\path{doi:10.1109/FOCS54457.2022.00047}}.

\bibitem[MCW78]{mcw78}
George Markowsky, J.~Lawrence Carter, and Mark~N. Wegman.
\newblock Analysis of a universal class of hash functions.
\newblock In {\em Mathematical Foundations of Computer Science 1978}, volume~64 of {\em Lecture Notes in Computer Science}, pages 345--354. Springer, 1978.
\newblock \href {https://doi.org/10.1007/3-540-08921-7_82} {\path{doi:10.1007/3-540-08921-7_82}}.

\bibitem[MRRR14]{mrrr-fastloadbalance}
Raghu Meka, Omer Reingold, Guy~N. Rothblum, and Ron~D. Rothblum.
\newblock Fast pseudorandomness for independence and load balancing.
\newblock In {\em Proceedings of the 41st International Colloquium on Automata, Languages, and Programming (ICALP)}, volume 8572 of {\em Lecture Notes in Computer Science}, pages 859--870. Springer, 2014.
\newblock \href {https://doi.org/10.1007/978-3-662-43948-7_71} {\path{doi:10.1007/978-3-662-43948-7_71}}.

\bibitem[MU05]{mitzupfal}
Michael Mitzenmacher and Eli Upfal.
\newblock {\em Probability and Computing: Randomized Algorithms and Probabilistic Analysis}.
\newblock Cambridge University Press, Cambridge, 2005.
\newblock \href {https://doi.org/10.1017/CBO9780511813603} {\path{doi:10.1017/CBO9780511813603}}.

\bibitem[MV84]{mv81}
Kurt Mehlhorn and Uzi Vishkin.
\newblock Randomized and deterministic simulations of prams by parallel machines with restricted granularity of parallel memories.
\newblock {\em Acta Informatica}, 21(4):339--374, 1984.
\newblock \href {https://doi.org/10.1007/BF00264615} {\path{doi:10.1007/BF00264615}}.

\bibitem[PT12]{PT12-tabulation}
Mihai P\v{a}tra\c{s}cu and Mikkel Thorup.
\newblock The power of simple tabulation hashing.
\newblock {\em J. ACM}, 59(3), June 2012.
\newblock \href {https://doi.org/10.1145/2220357.2220361} {\path{doi:10.1145/2220357.2220361}}.

\bibitem[RS98]{RS98-ballsbinstightanalysis}
Martin Raab and Angelika Steger.
\newblock Balls into bins -- a simple and tight analysis.
\newblock In {\em Randomization and Approximation Techniques in Computer Science (RANDOM 1998)}, volume 1518 of {\em Lecture Notes in Computer Science}, pages 159--170. Springer, 1998.
\newblock \href {https://doi.org/10.1007/3-540-49543-6_13} {\path{doi:10.1007/3-540-49543-6_13}}.

\bibitem[WC81]{WC81}
Mark~N. Wegman and J.~Lawrence Carter.
\newblock New hash functions and their use in authentication and set equality.
\newblock {\em Journal of Computer and System Sciences}, 22(3):265--279, 1981.
\newblock \href {https://doi.org/10.1016/0022-0000(81)90033-7} {\path{doi:10.1016/0022-0000(81)90033-7}}.

\bibitem[Woo14]{wootters2014thesis}
Mary~Katherine Wootters.
\newblock {\em Any Errors in This Dissertation Are Probably Fixable: Topics in Probability and Error Correcting Codes}.
\newblock Ph.d. dissertation, University of Michigan, 2014.
\newblock URL: \url{https://deepblue.lib.umich.edu/handle/2027.42/108844}.

\bibitem[WZ96]{WinklerZuckerman1996}
Peter Winkler and David Zuckerman.
\newblock Multiple cover time.
\newblock {\em Random Structures \& Algorithms}, 9(4):403--411, 1996.
\newblock \href {https://doi.org/10.1002/(SICI)1098-2418(199612)9:4<403::AID-RSA4>3.0.CO;2-0} {\path{doi:10.1002/(SICI)1098-2418(199612)9:4<403::AID-RSA4>3.0.CO;2-0}}.

\end{thebibliography}

\appendix
\section{Deferred Proofs}
%\begin{corollary}
%    Let $u\ge \ell$ and $c>0$ be a constant. For any $\eps < 1$ such that $\eps \ge 1/2^{u-\ell}$, there exists $C_\eps = O(1/\eps)$ such that the following holds. Let $S\subset\F_2^{u}$ be of cardinality $m$, and let $\cH$ be the uniform distribution over all surjective linear maps $h: \F_2^{u}\to\F_2^\ell$.  Then \[\Pr_{h\sim \cH}\left[M(S,h) < C_\eps\cdot \opt(m,n)\right] \ge 1-\eps.\]
%\end{corollary}
\subsection{Proof of \Cref{thm:quadratic-tail-final}: Removing the Surjectivity Assumption}
\label{app:removing-surjectivity}
%We show that with only a mild loss in parameters, we can use our max-load theorem on random surjective linear maps with large enough universe size (\Cref{lem:quad-tail-surj}) to prove max-load bounds on random linear maps with any universe size (\Cref{thm:quadratic-tail-final}). The intuition is that with extremely high probability, a random linear map is surjective, and that a smaller universe only makes hashing easier.

\begin{proof}[Proof of \Cref{thm:quadratic-tail-final}]
 We will first show the result for $u$ large enough, i.e. $2^{u-\ell}\ge \max\{(m-2)^2,(m\ell)^2\}$. At the end, we will remove this assumption on $u$. Let $\cH$ be the set of linear maps $\F_2^u\to\F_2^\ell$, and let $\cK$ be the distribution of the nullity of $h\sim\cH$. If $k$ is the nullity of $h$, then the number of universe elements $v\in \F_2^u$ such that $h(v) = 0$ is $2^k$. For $h\sim \cH$ we have $\Pr[h(v) = 0] = 2^{-\ell}$ for fixed $v\neq 0$. Hence, we can compute by linearity that \begin{align*}\E_{k\sim K}[2^k]  = \E_{h\sim \cH}\left[\sum_{v\in \F_2^u}\mathbf{1}\{h(v) = 0\}\right]    = \sum_{v\in \F_2^u} \Pr_{h}[h(v) = 0]   = 1 + 2^{-\ell}(2^u-1)   \le  2^{u-\ell} + 1.   \end{align*}
Let $\cE$ be the event that $h$ is surjective, i.e. $k = u-\ell$. By Markov's inequality on the random variable $2^k - 2^{u-\ell}$ (as $k\ge u-\ell$), we have \[ \Pr\left[\neg \cE\right]  = \Pr_{k\sim \cK}[k\ge u-\ell + 1] = \Pr_{k\sim \cK}\left[2^k - 2^{u-\ell} \ge 2^{u-\ell}\right] \le \frac{1}{2^{u-\ell}}\le \frac{1}{(m-2)^2} \le \frac{1}{(r-2)^2}.\]

 For brevity, set $M\coloneqq r\cdot \opt(m,n)$. We can use \Cref{lem:quad-tail-surj}, the above observation, and the fact $u\ge \ell + 2\log(\ell m)$ to deduce
 \begin{align*}
    \Pr_{h\sim \cH}\left[M(S,h) \ge M\right]    \le \Pr[\neg \cE] + \Pr_{h\sim \cH}\left[M(S,h) \ge M| \cE\right]  \le  \frac{1}{(r-2)^2} + \frac{48}{(r-2)^2}    =\frac{49}{(r-2)^2}.
 \end{align*}

 We will now remove the lower bound assumption on $u$. Intuitively, for any $\ell\le u' < u$ we can simply embed $\F_2^{u'}$ into $\F_2^u$, and then port in the max-load result for $\F_2^u$.

 More formally, let $\cH'$ be the set of linear maps $\F_2^{u'}\to\F_2^\ell$. Take an arbitrary $u'$-dimensional subspace $V\leq \F_2^u$. There is an isomorphism $\iota: \F_2^{u'}\to V$. Denote $T\coloneqq \iota(S)$. Since $T\subset V$, we have for any $h:\F_2^u\to\F_2^\ell$ that \begin{align*}M(T,h|_V)  = \max_{y\in \F_2^\ell} |(h|_V)^{-1}(y)\cap T|  = \max_{y\in \F_2^\ell} |h^{-1}(y)\cap V \cap T|  = \max_{y\in \F_2^\ell} |h^{-1}(y)\cap T| =M(T,h). \end{align*} For $h\sim \cH$, $h|_V:V\to \F_2^\ell$ will be a uniform random linear map. Consequently, \begin{align*}\Pr_{h'\sim \cH'}[M(S,h')\ge M] = \Pr_{h\sim \cH}[M(T, h|_V)\ge M] = \Pr_{h\sim \cH}[M(T, h)\ge M]  \le \frac{49}{(r-2)^2} \end{align*} since $T$ has cardinality $m$.
\end{proof}

%\begin{proof}[Proof Sketch]
    %We use the probabilistic method. Fix a linear map $h$, whose preimages will divide $\F_2^{u}$ into at most $n$ different bins. We now consider picking $S$ randomly, where each universe element in $\F_2^{u}$ is included in $S$ independently with probability $1/m$. The expectation and variance of $|S|$ can be computed to be $m$ and $m-1$, respectively. Hence, by Chebyshev's inequality, it follows that $\Pr[|S|\ge m - 2\sqrt{m-1}]\le 1/4$.

 %   Now notice that the load distribution of $S$ in the preimages of $h$ is \emph{exactly} the same distribution as randomly mapping $m$ balls to $n$ bins. From here, classical techniques finish. In the case $m\le \frac{1}2n\log n$, we can use Poisson Approximation theorem to bound $M(S,h)> (1-o(1)) \opt(m,n)$ by $2^{-3\log^2 n}$, and for $m > \frac{1}2n\log n$, we can use the Chernoff bound for negatively correlated variables to prove $M(S,h) > (1+o(1))\frac{m}n$ with probability. In either case, a union bound over all $n^{\log m}$ linear maps yields the desired result.
%\end{proof}

\end{document}